\def\compileForArxiv{}

\ifdefined\compileForArxiv
	\newcommand{\arxiv}[1]{#1}
	\newcommand{\camready}[1]{}
	\documentclass[runningheads]{llncs}
\else
	\newcommand{\arxiv}[1]{}
	\newcommand{\camready}[1]{#1}
	\documentclass{llncs}
\fi

\usepackage{booktabs}
\usepackage{hyperref}
\usepackage{amssymb, amsmath}
\usepackage{graphics,graphicx,color}
\usepackage{paralist}
\usepackage{bigdelim}
\usepackage{xspace}
\usepackage{array}
\usepackage{tabularx}
\usepackage{ifluatex}
\ifluatex
  \usepackage{fontspec}
  \defaultfontfeatures{Ligatures=TeX}
  \usepackage[positioning=sLeaderNorthEastBelowStacks,leadertype=sBezier,splitting=weightedMedian]{luatodonotes}
\else
  \usepackage[utf8]{inputenc}
  \usepackage[T1]{fontenc}
  \usepackage{todonotes}
\fi

\graphicspath{{pic/}}

\usepackage{tikz}
\usetikzlibrary{matrix,shapes,calc,arrows,shapes.geometric,positioning}
\usepackage[nomessages]{fp}

\newenvironment{tikzreduction}[3][f]{
 \begin{tikzpicture}
   \def\vis{#1}
   ;
   \def\c{cells};
   
   \def\w{#2}
   \def\h{#3}
   
      \tikzset{whitenode/.style={
    circle,fill=white,inner sep=2pt,draw=black
  }
 }
    \tikzset{blacknode/.style={
    circle,fill=black,inner sep=2pt
  }
 }
   
   \tikzset{move/.style={
    ->,bend angle=90, bend right=30, red, shorten >=.2ex, shorten <=.2ex
  }
 }
   \tikzset{moveA/.style={
    ->,bend angle=90, bend right=30, blue, shorten >=.2ex, shorten <=.2ex
  }
 }
 \tikzset{moveB/.style={
    ->,bend angle=90, bend right=30, orange, shorten >=.2ex, shorten <=.2ex
  }
 }
  \tikzset{moveC/.style={
    ->,bend angle=90, bend right=30, green, shorten >=.2ex, shorten <=.2ex
  }
 }
 \ifx\vis\c
    \foreach \x in {0,...,\h} {
      \draw [gray,dashed]($(-0.5,-0.5)+(0,\x)$)node[](){} -- ($(-0.5,-0.5)+(\w,\x)$)node[](){};
    };
    \foreach \y in {0,...,\w} {
      \draw [gray,dashed]($(-0.5,-0.5)+(\y,0)$)node[](){} -- ($(-0.5,-0.5)+(\y,\h)$)node[](){};
    };
 \else
    \foreach \x in {0,...,\h} {
      \draw [gray,dashed]($(0,\x)$)node[](){} -- ($(\w,\x)$)node[](){};
    };
    \foreach \y in {0,...,\w} {
      \draw [gray,dashed]($(\y,0)$)node[](){} -- ($(\y,\h)$)node[](){};
    }; 
 \fi
}{\end{tikzpicture}}

\newenvironment{sketchofproof}{\begin{proof}[sketch]}{\qed\end{proof}}

\newcommand{\NP}{\ensuremath{\mathcal{N\!P}}}
\newcommand{\APX}{\ensuremath{\mathcal{AP\!X}}}
\newcommand{\lang}[1]{\textsc{#1}}

\begin{document}
\title{Snapping Graph Drawings to the Grid Optimally}

\author{Andre L\"offler \and Thomas C. van Dijk \and Alexander Wolff}

\authorrunning{A.~L\"offler et al.}

\tocauthor{A.~L\"offler, T.~C.~van~Dijk, A.~Wolff}

\institute{%
    Lehrstuhl f\"ur Informatik I, Universit\"at W\"urzburg, Germany.\\
    \email{http://www1.informatik.uni-wuerzburg.de/en/staff}
}
\maketitle

\begin{abstract}
  In geographic information systems and in the production of digital
  maps for small devices with restricted computational resources one
  often wants to round coordinates to a rougher grid.  This removes
  unnecessary detail and reduces space consumption as well as
  computation time.  This process is called \emph{snapping to the grid} 
  and has been investigated thoroughly from a computational-geometry
  perspective.  In this paper we investigate the same problem for
  given drawings of planar graphs under the restriction that their
  combinatorial embedding must be kept and edges are drawn
  straight-line.  We show that the problem is NP-hard for several
  objectives and provide an integer linear programming formulation.
  Given a plane graph~$G$ and a positive integer~$w$, our ILP
  can also be used to draw $G$ straight-line on a grid of width~$w$
  and minimum height (if possible).
\end{abstract}

\section{Introduction}\label{sec:intro}

When compressing geographic data, for example in order to ship it to
devices with small memory, small screens and slow CPUs, the main
objective is to reduce unnecessary detail.  One way to do this is to
round data points to a grid.  

In the computational geometry community, a process called \emph{snap
rounding} has been proposed and has since become well-established:
given an arrangement of line segments, each grid cell that
contains vertices or intersections is ``hot''.
Then every segment becomes a polygonal chain whose edges (\emph{fragments})
connect center points of hot cells, namely those that the
original segment (\emph{ursegment})
intersects. 
Guibas and Marimont~\cite{gm-rad-IJCGA98} showed that during snap
rounding, vertices of the arrangement never cross a polygonal chain,
so after snapping no two fragments cross.
Moreover, the circular order of the fragments around an
output vertex is the same as the order in which the
corresponding ursegments intersect the boundary of its grid cell. 
The
resulting arrangement approximates the original one in the sense that
any fragment lies within the Minkowski sum of the corresponding
ursegments and a unit square centered at the origin.  
However, the structure of the graph can be affected (vertices merge, faces
disappear, edges bend). 
Further work in this direction includes that of De Berg et al.~\cite{de2007intersection}.

Motivated by the above GIS application, we investigate the problem of moving the drawing of a graph to a given grid.  
Since we still want to be able to recognize the original graph, we do not tolerate new incidences.  
Then we must accept the possibility that a vertex does not go to the nearest grid point, but we still want to minimize change.
This can by measured, for example, by the sum of the distances or the maximum distance in the Euclidean ($L_2$-) or Manhattan ($L_1$-) metric.\
Apparently, this problem, which we call \lang{Topologically-Safe
  Snapping}, has not been studied yet.  (Note that we carry over the
term ``snapping,'' although we don't necessarily snap to the
\emph{nearest} grid point.)

From a graph-drawing perspective, restricting to the grid has a (relatively) long history.  
Motivated by the fact that Tutte's barycenter method~\cite{t-hdg-PLMS63} for drawing planar graphs yields drawings that need precision linear in the size of the graph, Schnyder~\cite{s-epgg-SODA90} and, independently, de Fraysseix et al.~\cite{fraysseix1990draw} have shown that any planar graph with  $n$ vertices admits a straight-line drawing on a grid of size $O(n)\times O(n)$.  
This is asymptotically optimal in the worst case~\cite{fraysseix1990draw}. 
Chrobak and Nakano~\cite{chrobak1998minimum} have investigated drawing planar graphs on grids of smaller width, at the expense of a larger height.
Grid-snapping techniques can be found in any diagram creation tool.
Aesthetic properties of force-directed drawing algorithms are widely
researched, see e.g. Kieffer et al.~\cite{kieffer2013incremental} for grid layouts of diagrams. 

Although minimizing the area of straight-line grid drawings has been the topic of several graph drawing contests, 
there has been rather little previous work.  
It is known that the problem is \NP{}-hard~\cite{kw-mapslgd-gd07}, but not even for special cases exact or approximation algorithms have been proposed.

\paragraph{Our contribution.}
We show that optimal snapping is \NP{}-hard, with a reduction that
asks for compressing each coordinate by just a single bit (Sect.~\ref{sec:hardness}).
The proof is somewhat similar in concept
to the proof of the \NP{}-hardness of Metro-Map
Layout~\cite{n-admm-MTh05,w-dsms-IFE07}, but new constructions are
required since the snapping problem does not easily allow the
construction of ``rigid'' gadgets. 
Second, we give an integer linear program (ILP) for optimal snapping
(Sect.~\ref{sec:ilp}).
This ILP generalizes the one for Metro-Map
Layout~\cite{nw-dlhqm-TVCG11}.  Where that ILP assumes a constant number of
possible edge
directions (namely~8), we have to cope with a number that is quadratic
in the size of the grid. 
The numbers of variables and constraints of our ILP are polynomial in
grid and graph size, but are quite large in practice. 
In fact, on a grid of size $k \times k$, there are $\Theta(k^2)$ edge
directions.  Thus, for an $n$-vertex planar graph, we must generate
$O(k^2n^2)$ constraints, among others, to preserve planarity and the
cyclic order of edges around the vertices.
To ameliorate this, we apply delayed constraint generation, a technique
that adds certain constraints only when needed.  Still, runtime is
prohibitive for graphs with more than about 15 vertices.
Our techniques can be adapted to draw (small) graphs with minimal area.
This is interesting even for small graphs since minimum-area drawings
can be useful for validating (counter)examples in graph drawing theory.

\section{NP-Hardness}
\label{sec:hardness}

We start with a formal definition of \lang{TopologicallySafeSnapping} -- or \lang{TSS} for short. 
To measure the cost of rounding a graph, we utilize Manhattan distance
and the total cost of rounding a graph is the sum over the individual
costs of the vertices.
As input we take a plane graph $G=(V,E)$ with vertex positions and a
bounding box $[0, X_{\max}] \times [0,Y_{\max}]$.
The \lang{TSS} problem is then to minimize the cost of rounding the
vertices of~$G$ to the integer grid within the box without altering
the topology with respect to the given plane straight-line drawing of~$G$.

We prove \NP{}-hardness of \lang{TSS} by considering the decision variant: 
is there a rounding that does not exceed a given cost bound $c$? 
We reduce from \lang{Planar monotone 3-SAT} (which is
\NP{}-hard~\cite{de2012optimal}): given a formula $F$ in 3-CNF that
is monotone and whose graph $H(F)$ is planar, is $F$ satisfiable?
The graph $H(F)$ has a vertex for each variable and each clause of~$F$
and an edge between a variable vertex~$v_X$ and a clause vertex~$v_C$
if $X$ is part of~$C$.
We will only consider formulae whose graphs are planar and that are
\emph{monotone} in the usual sense:
for any clause~$C$, variables in~$C$ either are all negated or all unnegated.
We can assume that the graph $H(F)$ can be laid out as in
Fig.~\ref{figure:graph-line-bend-variable-clause}$\,$(a): all variable
vertices lie on the x-axis, the vertices of all-negated clauses
lie above the x-axis, and the vertices of all-unnegated clauses lie
below the x-axis~\cite{de2012optimal}.

\begin{theorem}
\lang{TopologicallySafeSnapping} is \NP{}-hard.
\end{theorem}

\begin{proof}
For a given monotone, planar 3-CNF formula $F$, we construct a cost bound $c_{\min}$ and a plane graph $G$ with vertices at half-integer coordinates.
The sum of all vertex movements induced by rounding $G$ to integer coordinates is exactly $c_{\min}$ if and only if $F$ is satisfiable.
To achieve this, we introduce gadgets for the elements of $H(F)$ -- variables, clauses, edges and bends -- and construct $G$ and $c_{\min}$ in polynomial time. 

For exposition, we consider two types of vertices.
Black vertices start on integer grid points and do not need to be rounded.
Moving a black vertex to another integer grid point is allowed, but we will show that this is not optimal if  $F$ is satisfiable.
White vertices start at grid cell centers and thus will always move at
least one unit by rounding.  Let~$W \subseteq V(G)$ be the set of 
white vertices.  Now we give the construction of the various gadgets.

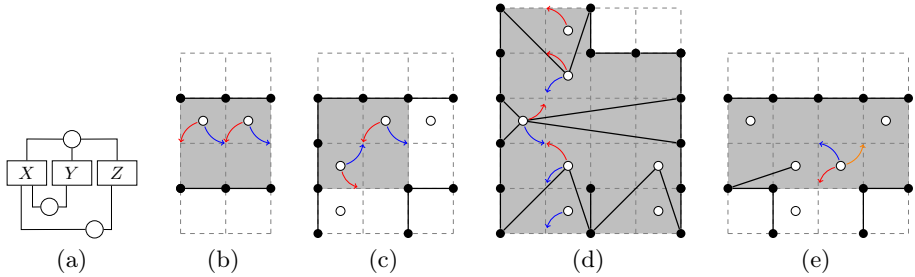
\begin{figure}[t]
  \centering
  \begin{tabularx}{\textwidth}{@{}cXcXcXcXc@{}}
    \scalebox{.75}{
      \begin{tikzpicture} 
        \draw (0,2)node[rectangle, minimum width=.7cm, draw](x){$X$}
        (.8,2)node[rectangle, minimum width=.7cm, draw](y){$Y$}
        (1.6,2)node[rectangle, minimum width=.7cm, draw](z){$Z$}; 
        \draw (0.8,2.6)node[circle, minimum width=.2cm, draw](c1){}
        (0.4,1.4)node[circle, minimum width=.2cm, draw](c2){}
        (1.2,1)node[circle, minimum width=.2cm, draw](c3){}; 
        \draw ($(x.south)+(0.1,0)$) -- ($(x.south)+(0.1,-0.38)$) -- (c2) -- ($(y.south)+(-0.1,-0.38)$) -- ($(y.south)-(0.1,0)$); 
        \draw ($(x.south)-(0.1,0)$) -- ($(x.south)-(0.1,0.78)$) -- (c3) -- ($(z.south)+(-0.1,-0.78)$) -- ($(z.south)-(0.1,0)$); 
        \draw (x.north) -- ($(x.north)+(0,0.38)$) -- (c1) -- ($(z.north)+(0,0.38)$) -- (z.north); 
        \draw (y.north)--(c1);
      \end{tikzpicture}
    }
    &&
    \scalebox{.6}{
      \begin{tikzreduction}[]{2}{4} 
        \draw[fill=gray!100, draw=none, opacity=.5] (0,1) -- (2,1) -- (2,3) -- (0,3) -- cycle;
        \draw (2,3) node[blacknode](2,3){} (0,3)
        node[blacknode](0,3){} (0,1) node[blacknode](0,1){} (1,3)
        node[blacknode](1,3){} (2,1) node[blacknode](2,1){}; 
        \draw (1.5,2.5) node[whitenode](line2){} (0,3)
        node[blacknode](0,3){} (0.5,2.5) node[whitenode](line){} (0,1)
        node[blacknode](0,1){} (1,1) node[blacknode](1,1){};
        \draw [thick] (0,3)to(0,3)to(1,3)to(2,3);
        \draw [thick] (0,1)to(0,1)to(1,1)to(2,1);
        \draw [move] (line) to (0,2);
        \draw [moveA] (line) to (1,2);
        \draw [move] (line2) to (1,2);
        \draw [moveA] (line2) to (2,2);
      \end{tikzreduction}
    }
    &&
    \scalebox{.6}{
      \begin{tikzreduction}[]{3}{4} 
        \draw[fill=gray!100, draw=none, opacity=.5] (0,1) -- (2,1) -- (2,3) -- (0,3) -- cycle;
        \draw (0,0) node[blacknode](0,0){} (0,1) node[blacknode](){} (2,1) node[blacknode](){}
        (0,3) node[blacknode](){} (1,3) node[blacknode](){} (0,2)
        node[blacknode](){} (2,3) node[blacknode](){} (0.5,1.5)
        node[whitenode](bend1){} (0.5,0.5) node[whitenode](){}  (2.5,2.5) node[whitenode](){} (1.5,2.5) node[whitenode](bend2){}  (3,3) node[blacknode](){} (2,0) node[blacknode](){} (3,1) node[blacknode](){}; 
        \draw [thick] (0,0)to(0,1)to(0,2)to(0,3)to(1,3)to(2,3)to(3,3);
        \draw [thick] (2,0)to(2,1)to(3,1);
        \draw [move] (bend1) to (1,1);
        \draw [move] (bend2) to (1,2);
        \draw [moveA] (bend1) to (1,2);
        \draw [moveA] (bend2) to (2,2);
      \end{tikzreduction}
    }
    &&
    \scalebox{.6}{
      \begin{tikzreduction}[]{4}{5} 
      \draw[fill=gray!100, draw=none, opacity=.5] (0,0) -- (0,5) -- (2,5) -- (2,4) -- (4,4) -- (4,0) -- cycle;
        \draw (0,0) node[blacknode](1ien73ttadpo8){} (3,4)
        node[blacknode](1hv2p7k71h8bn){} (1.5,4.5)
        node[whitenode](t1){} (4,2) node[blacknode](r1var1){}; 
        \draw (2,1) node[blacknode](uvepc2luaxcy){} (3.5,0.5)
        node[whitenode](1iusntk2tbrjr){} (4,1)
        node[blacknode](soknf1udh3jn){} (4,1)
        node[blacknode](rbottvar1){}; 
        \draw (4,0) node[blacknode](us09poqnlt2a){} (2,5)
        node[blacknode](1hrutdrf1se9l){} (0,2)
        node[blacknode](l1var1){} (4,3) node[blacknode](r2var1){}; 
        \draw (0,1) node[blacknode](lbottvar1){} (4,4)
        node[blacknode](rtoppvar1){} (4,4)
        node[blacknode](1jet99xtmpz6p){};
        \draw (1.5,0.5) node[whitenode](b1){} (1.5,1.5)
        node[whitenode](b2){} (2,4) node[blacknode](sl3geuv35qhu){}
        (2,0) node[blacknode](ubj5489w47qs){};
        \draw (0.5,2.5) node[whitenode](decider){} (1.5,3.5)
        node[whitenode](t2){} (0,3) node[blacknode](l2var1){} (2,4)
        node[blacknode](vb9lq5yx6tlu){}; 
        \draw (0,4) node[blacknode](ltoppvar1){} (3.5,1.5)
        node[whitenode](trqv20r71eu9){} (0,5)
        node[blacknode](s1018e5yusky){}; 
        \draw [thick] (lbottvar1)to(l1var1)to(l2var1)to(ltoppvar1);
        \draw [thick] (1ien73ttadpo8)to(lbottvar1);
        \draw [thick] (uvepc2luaxcy)to(ubj5489w47qs)to(b2)to(1ien73ttadpo8);
        \draw [thick] (s1018e5yusky)to(ltoppvar1);
        \draw [thick] (vb9lq5yx6tlu)to(1hrutdrf1se9l)to(t2)to(s1018e5yusky);
        \draw [thick] (soknf1udh3jn)to(us09poqnlt2a)to(trqv20r71eu9)to(ubj5489w47qs); 
        \draw [thick] (1jet99xtmpz6p)to(1hv2p7k71h8bn)to(sl3geuv35qhu);
        \draw [thick] (rbottvar1)to(r1var1)to(r2var1)to(rtoppvar1);
        \draw [thick] (l1var1)to(decider)to(l2var1);
        \draw [thick] (r1var1)to(decider)to(r2var1);
        \draw [move] (decider) to (1,3);
        \draw [move] (t1) to (1,5);
        \draw [move] (t2) to (1,4);
        \draw [moveA] (t2) to (1,3);
        \draw [moveA] (decider) to (1,2);
        \draw [moveA] (b1) to (1,0);
        \draw [moveA] (b2) to (1,1);
        \draw [move] (b2) to (1,2);
      \end{tikzreduction}
    }
    &&
    \scalebox{.6}{
      \begin{tikzreduction}[]{4}{4} 
        \draw[fill=gray!100, draw=none, opacity=.5] (0,1) -- (4,1) -- (4,3) -- (0,3) -- cycle;
        \draw (1.5,1.5) node[whitenode](c1wl){} (3,1)
        node[blacknode](3,1){} (3,3) node[blacknode](3,3){} (1,3)
        node[blacknode](1,3){} (0.5,2.5) node[whitenode](0.5,2.5){}
        (2,3) node[blacknode](2,3){} (4,3) node[blacknode](4,3){}
        (1,1) node[blacknode](1,1){} (1,0) node[blacknode](1,0){} (3,0) node[blacknode](3,0){} (2.5,1.5)
        node[whitenode](decider){} (3.5,2.5)
        node[whitenode](3.5,2.5){} (0,3) node[blacknode](0,3){} (0,1)
        node[blacknode](c1ankle){} (4,1) node[blacknode](4,1){} (1.5,0.5)
        node[whitenode](){};
        \draw [thick] (c1wl)to(c1ankle);
        \draw [thick] (c1ankle)to(1,1)to(1,0);
        \draw [thick] (3,0)to(3,1)to(4,1);
        \draw [thick] (0,3)to(1,3)to(2,3)to(3,3)to(4,3);
        \draw [move] (decider) to (2,1);
        \draw [moveA] (decider) to (2,2);
        \draw [moveB] (decider) to (3,2);
      \end{tikzreduction}
    }
    \\
    (a)&&(b)&&(c)&&(d)&&(e)
  \end{tabularx}
  \caption{(a) Graph $H(F)$ for formula $F = (\overline{X} \vee
    \overline{Y} \vee \overline{Z}) \wedge (X \vee Y) \wedge (X \vee
    Z)$, (b)~horizontal line gadget, (c)~bottom-to-right bend gadget,
    both with possible roundings, (d)~gadget for variable with two
    negated and one unnegated occurrences, (e)~all-negated clause gadget with
    three negated variables. Inner area of each gadget highlighted gray.}
  \label{figure:graph-line-bend-variable-clause}
\end{figure}

First, we introduce the line and bend gadgets.
These ensure consistency between variable and clause gadgets.
Every segment of the line gadget consists of four black vertices and
two edges forming a \emph{tunnel}, and a single white vertex inside;
see Fig.~\ref{figure:graph-line-bend-variable-clause}$\,$(b). 
The white vertex can be rounded most cheaply to exactly two possible integer grid points, depicted by the red and blue arrows.
By rounding a white vertex in one direction, we prohibit the neighbor in that direction to go the opposite way -- as both vertices would end up on the same integer grid point (which violates topological safety).
So, if the white vertex at one end of the line is rounded inward (blue arrow) the white vertex at the other end of that line must be rounded outward -- we say it is \emph{pushed}. 
The same holds for the bend gadgets, as can be seen in Fig.~\ref{figure:graph-line-bend-variable-clause}$\,$(c).

Next, consider the variable gadget depicted in Fig.~\ref{figure:graph-line-bend-variable-clause}$\,$(d). 
It has tunnels for vertical line gadgets for every negated and unnegated occurrence at the top and bottom respectively. 
At the center of this gadget, there is a white vertex that is connected to the gadget's walls by two triangles. 
Call this the \emph{assignment} vertex and note that it can be rounded
up or down, which makes the edges of the triangles block grid
points on the top or bottom tunnels, respectively.
The tunnels of that direction are then all forced to push into the
connected clause gadgets.
This represents the truth assignment of the corresponding variable.

Finally, the clause gadget is shown in
Fig.~\ref{figure:graph-line-bend-variable-clause}$\,$(e).
We describe the all-negated degree-3 version; the degree-2 version can be
constructed similarly.
There is a white \emph{satisfaction} vertex that can go to any of three possible integer grid points at equal cost.
These grid points belong to line gadgets and are only available if the line does not ``push.''
Then the satisfaction vertex can be rounded at cost $1$ if and only if the clause is satisfied.
Gadgets for all-unnegated clauses can be obtained by mirroring the
construction of
Fig.~\ref{figure:graph-line-bend-variable-clause}$\,$(e) at a
horizontal line.

The rounding cost of $G$ is bounded from below by $c_{\min} = |W|$
since every white vertex must be rounded at cost at least~$1$. 
If $F$ is satisfiable, there is a rounding that achieves this because
then we can round the assignment vertices such that the satisfaction
vertices can be rounded at cost~$1$. 
In the other direction, a satisfying assignment can be read off from the assignment vertices if rounding occurred at cost $c_{\min}$.

If none of the three candidate grid points for the satisfaction vertex are available, a topologically correct rounding must move a black vertex associated with that clause (of either the clause itself, the connected variables or the edges and bends connecting them).
This adds at least $1$ to the rounding cost without reducing the movement of any white vertex and thus such solutions cost strictly more than $c_{\min}$.
That is, if $c_{\min}$ is exceeded, then $F$ is unsatisfiable: any rounding corresponding to a satisfying truth assignment is cheaper.
This concludes our Karp reduction and the claim follows.
\qed
\end{proof}

\begin{corollary}
  \lang{TopologicallySafeSnapping} is also \NP-hard when using
  Euclidean distance.  In this case it is also \NP-hard to minimize
  the maximum movement instead of the sum.
\end{corollary}
\begin{sketchofproof}
  The above proof goes through with Euclidean distance and $c_{\min} =
  \sqrt{0.5^2 + 0.5^2} \cdot |W|$.  For minimizing the maximum
  movement, observe that rounding white vertices now costs less, but
  moving a black vertex still has cost at least~$1$: if $F$ is
  satisfiable, the maximum movement is $\sqrt{0.5^2 + 0.5^2}$,
  otherwise it is at least~$1$.
\end{sketchofproof}
This distinction of maximum movement ($\sqrt{0.5^2 + 0.5^2}\approx 0.71$ versus $1$) based on the satisfiability of $F$ also gives the following.

\begin{corollary}
 Euclidean \lang{TopologicallySafeSnapping} with the objective to minimize maximum movement is \APX{}-hard.
\end{corollary}

\section{Exact Solution using Integer Linear Programming}
\label{sec:ilp}

In this section we provide an ILP-based exact algorithm for \lang{TSS}.
Recall that an instance is a graph $G=(V,E)$ with vertex coordinates.
For all $v \in V$, call these $(X_v,Y_v)$ and introduce integer decision variables $0\leq x_v\leq X_{\max}$ and $0\leq y_v\leq Y_{\max}$ to represent the ``rounded'' output position.
This leads to the following objective function.
 \begin{equation}
 \mathrm{Minimize\ } \textstyle \sum_{v \in V} |x_v - X_v | + |y_v - Y_v |
 \end{equation}
This formula is itself not linear, but can be made so with standard transformations~\cite{lpBook}.
Note that without any further constraints, this would just move every vertex to the nearest integer grid point.
We will now introduce constraints to ensure topological safety, that is, in the output no two points are on same grid point, no two edges intersect, and the edges at every node have the same cyclic order as in the input.
 
\paragraph{Vertices do not coincide.}
This can be ensured by adding the following constraints.
They too are not linear as stated, but can be readily linearized.
 \begin{equation}
    (x_v \neq x_w) \lor (y_v \neq y_w) \quad\quad\forall v,w \in V, v\neq w
 \label{equation:constraint-overlap}
 \end{equation} 
\paragraph{Possible directions.}
The most important departure from the metro-map drawing ILP is that, clearly, more than eight different directions are allowed.
A priori we have no further constraints than that every rounded vertex lies somewhere within the given bounding box.
Let $\mathcal{D}$ be the set of unique directions $D = (D_X, D_Y)$ in $[-X_{\max}, X_{\max}] \times [-Y_{\max}, Y_{\max}]$.
Considering the Farey sequence~\cite{gram1994concrete}, we know that
$|\mathcal{D}|$ is $\Theta(X_{\max}\cdot Y_{\max})$.
In the following, we let the set $\mathcal{D}$ be ordered counterclockwise, starting at the positive x-axis, allowing comparison of directions.
 
\paragraph{No two edges cross.}
The following constraints ensure that nonincident edges do not cross.
(Incident edges are allowed to touch in the shared vertex.)
We will follow the idea of N{\"o}llenburg and Wolff~\cite{nw-dlhqm-TVCG11}.
While producing octilinear drawings of metro maps, they ensured planarity by forcing every pair of nonincident edges to be separated by at least some distance $D_{\min}$ in at least one of the eight octilinear directions.
This minimum distance was partly an aesthetic guideline, but also guarantees planarity.
We are only interested in the latter and therefore pick $D_{\min}$ such that all planar realizations on the grid are allowed.
 
The separation distance $D_{\min}$ has to be small enough to separate any non-intersecting pair of edges in the output.
Here the bounding box leads to a bound since it bounds the slope of the edges; 
it suffices to choose $D_{\min} = 1/(\max \{ X_{\max}, Y_{\max} \}+1)$.
 
For every pair of nonincident edges $e_1,e_2 \in E$ and every
direction $D \in \mathcal{D}$, we introduce a binary decision variable
$\gamma_D(e_1,e_2) \in \{0,1\}$ indicating that $e_1$ and $e_2$ are
apart by~$D_{\min}$ in direction~$D$.
Every such pair must be separated in some direction (following the idea of \cite{n-admm-MTh05}).
 \begin{equation}\label{equation:one_gamma_eq_one}
  \textstyle \sum_{D \in \mathcal{D}} \gamma_D(e_1,e_2) = 1 \quad\quad\forall e_1,e_2 \in E, e_1,e_2 \text{ nonincident}
 \end{equation}
Let $L_\gamma = 2\cdot \max \{X_{\max},Y_{\max}\} +1$.  Then,
for any direction $D \in \mathcal{D}$, any pair of nonincident edges
$e_1, e_2$ and any $v \in e_1, w \in e_2$, we require the following.
\begin{equation}
  \label{equation:separation}
  D_X \cdot (x_{v} - x_{w}) + D_Y \cdot (y_{v}- y_{w}) + (1 - \gamma_D(e_1,e_2) )L_\gamma \geq D_{\min}
\end{equation}
Constraint~(\ref{equation:one_gamma_eq_one}) yields a unique
direction~$D$ with $\gamma_D = 1$.  By choice of $L_\gamma$, any
constraint~(\ref{equation:separation}) that involves a direction~$D$
with $\gamma_D = 0$ is trivially fulfilled.

\paragraph{Determine direction of incident edges.}
For incident edges $e_1, e_2 \in E$, we have to ensure that the directions of $e_1$ and $e_2$ differ.
Again, we generalize the metro-map drawing ILP -- dropping the ``relative position rule'' -- allowing edges to have any direction $D \in \mathcal{D}$.

To keep track of this, we introduce a binary decision variable
$\alpha_D(v,w) \in \{0,1\}$ for every vertex $v \in V$, every neighbor
$w\in N(v)$ and every direction $D\in \mathcal{D}$.  The meaning of
$\alpha_D(v,w)=1$ is that the direction of edge $(v,w)$ is~$D$.
\begin{equation}\label{equation:one_alpha_eq_one}
  \textstyle \sum_{D \in \mathcal{D}} \alpha_D(v,w) = 1 \quad\quad\forall v \in V \; \forall w \in N(v) 
\end{equation}
For any vertex $v \in V$, any neighbor $w \in N(v)$, and any direction
$D \in \mathcal{D}$, the following ensures that edge~$(v,w)$ indeed
has direction~$D$.  Let $L_\alpha = 2\cdot \max \{X_{\max},Y_{\max}\}+1$.
\begin{equation}
  \begin{aligned}
    x_w\cdot D_Y + y_v \cdot D_X - x_v \cdot D_Y \pm (1-\alpha_D(v,w)) L_\alpha & \gtreqless y_w \cdot D_X\\ 
    (1- \alpha_D(v,w)) L_\alpha + (x_w-x_v) \cdot D_X + (y_w-y_v) \cdot D_Y & \geq 0
  \end{aligned}
  \label{equation:alpha-assign}
 \end{equation}
From constraint (\ref{equation:one_alpha_eq_one}) we get that for every vertex-neighbor pair one $\alpha$ has to be set to $1$.
This $\alpha$ again enables one subset---as $L_\alpha$ dominates all
other terms---of constraints from (\ref{equation:alpha-assign}),
forcing comparison between edge slope and direction.  
This gives us the direction of edge $(v,w)$ with the correct sign.
 
\paragraph{Preserve cyclic order of outgoing edges.}
We use a binary decision variable $\beta(v,w) \in \{0,1\}$ for every vertex-neighbor pair, indicating if $w$ is the ``last'' neighbor of $v$ according to the order of $\mathcal{D}$.
The following preserves cyclic order.

\begin{equation}
  \label{equation:one_beta_eq_one}
		\textstyle \sum_{w \in N(v)} \beta(v,w) = 1
   \quad\quad\forall v \in V \text{ with } \deg (v) > 1 \hspace{4ex}
\end{equation}
\begin{equation}
  \begin{aligned}
    \hspace{2.8ex}\alpha_{D_1}&(v,w_i) \le \;\beta(v,w_i) + \textstyle 
    \sum_{D_w \in \mathcal{D} \colon D_w > D_1} \alpha_{D_w}(v,w_{i+1})\\ 
    &  \;\forall D_1 \in \mathcal{D} \; \forall v \in V, N(v) =
    \{w_1,w_2,\dots,w_k\} \; (k = \deg v > 1) 
  \end{aligned}
  \label{equation:unique-direction}
 \end{equation}
For notational convenience, we let $w_{k+1} = w_1$, as $N(v)$ is conceptually circular.
For any $\alpha$ set to $0$, the inequalities of (\ref{equation:unique-direction}) are trivially satisfied. 
Otherwise, there has to be a neighbor whose connecting edge has a later direction (and thus the corresponding $\alpha$ set to $1$), unless it is the last neighbor in the embedding of $v$.
To ensure that there is only one ``last neighbor''-violation of the constraints from (\ref{equation:unique-direction}), we introduce the constraints of (\ref{equation:one_beta_eq_one}).
Adding $\beta$ to every constraint of (\ref{equation:unique-direction}) also allows for the whole neighborhood of $v$ to be rotated around it.
This describes the full ILP and gives to the following.

\begin{theorem}
  The above ILP solves \lang{TopologicallySafeSnapping}.
\end{theorem}

\paragraph{Graph drawing.}
Replacing the objective function with $\mathrm{Minimize\ } \max_{v\in V} y_v$, the ILP computes a straight-line grid drawing with the given embedding, width at most $X_{\max}$, and minimum height.
This allows us to find minimum-area drawings of small graphs.

\paragraph{Delayed constraint generation.} 

We can apply a delayed constraint generation approach (see for example Cinneck \cite{Chinneck}) to the above ILP as follows.
First we run the ILP without any constraints, which snaps each vertex to the nearest grid point. (This takes practically no time.)
We then test the result for topological validity, adding constraints corresponding to any violations.
Then we repeat until no violations occur.
This improves the runtime when few iterations suffice for a particular instance, but the approach should still be considered practically infeasible, especially for large bounding boxes: the set of possible directions $\mathcal{D}$ still results in a large program.
Future work could focus on reducing the brute-force inclusion of all possible directions.
\arxiv{Experimental results are found in the appendix.}
\camready{Experimental results are found in the full version of this
  paper~\cite{ldw-sgdgo-arXiv16}.}

\bigskip
\noindent
\textbf{Acknowledgments.}  We thank Gergely Mincsovics for suggesting
this problem to us.

\bibliographystyle{abbrv}
\bibliography{abbrv,gd}

\ifdefined\compileForArxiv

\appendix
\newpage
\section*{Appendix: Experimental evaluation}
In the following section, we will discuss performance of an IBM CPLEX implementation of the above ILP on graphs different in vertex count, size of bounding rectangle and number of ``difficult'' parts.
We ran experiments on a Linux machine with 16 cores (2666 MHz and 4 MB cache each), 16 GB memory and 20 GB swap space and using the Java bindings for CPlex.
Model size is measured by considering the number of rows and columns before and after CPLEX completes any preprocessing step, ``*'' means that no model could be created within given time and memory.
Any times given are in wall clock time, and ``\dag'' means that there was no result within 10 minutes of computation.

In the following, \emph{full model} is used for executions of the above ILP without row generation. 
The column \emph{first} gives the time until any feasible integer solution (not necessarily optimal) is reported by the integer solver. 
In both cases, \emph{optimal} gives the time until the solver reports an optimal solution.
For the output figures, white vertices represent the initial positions with the red arrows indicating actual vertex movement.

We start with a rather small graph (Fig.~\ref{example:graph_1}).
Because of its size, building the model and finding the solution is
quick.  However, its vertices are positioned so that many constraints
are not trivially satisfied and significant effort is required even by
the row generation approach.

\begin{figure}[hb]
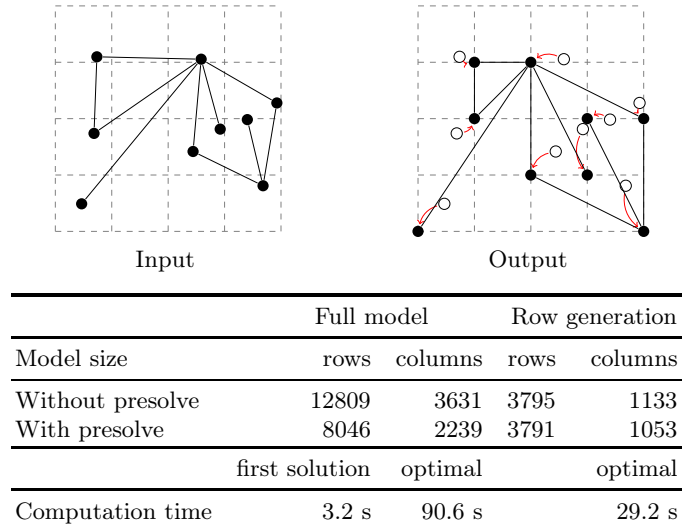

\centering
\begin{tabular}{cm{.1\textwidth}c}
\scalebox{.75}{
\begin{tikzreduction}{4}{4}
 \draw (0.46643,0.488909)node[blacknode](soejj3tqr6g4){}; 
\draw (2.58481,3.05461)node[blacknode](t8a6h7ishmwp){}; 
\draw (0.7365,3.0968)node[blacknode](qhqlvtjq8lrt){}; 
\draw (0.685863,1.738)node[blacknode](s162pejbsc37){}; 
\draw (2.44134,1.41729)node[blacknode](1ib393v8t30ab){}; 
\draw (3.68199,0.809619)node[blacknode](1jylmdvs4j12x){}; 
\draw (3.92674,2.27814)node[blacknode](rl7xs60bs7qt){}; 
\draw (2.92241,1.81396)node[blacknode](rhai685m0i7o){}; 
\draw (3.40347,1.98275)node[blacknode](vf9qpmewlqhv){}; 
\draw (t8a6h7ishmwp) -- (soejj3tqr6g4); 
\draw (qhqlvtjq8lrt) -- (t8a6h7ishmwp); 
\draw (s162pejbsc37) -- (qhqlvtjq8lrt); 
\draw (t8a6h7ishmwp) -- (s162pejbsc37); 
\draw (1ib393v8t30ab) -- (t8a6h7ishmwp); 
\draw (1jylmdvs4j12x) -- (1ib393v8t30ab); 
\draw (rl7xs60bs7qt) -- (1jylmdvs4j12x); 
\draw (t8a6h7ishmwp) -- (rl7xs60bs7qt); 
\draw (rhai685m0i7o) -- (t8a6h7ishmwp); 
\draw (vf9qpmewlqhv) -- (1jylmdvs4j12x); 
\end{tikzreduction}
}
 & &
\scalebox{.75}{
 \begin{tikzreduction}{4}{4}
\draw (0.0,0.0)node[blacknode](to2amp8wb2as){}; 
\draw (0.46643,0.488909)node[whitenode](oto2amp8wb2as){};
\draw [move] (oto2amp8wb2as) to (to2amp8wb2as);
\draw (2.0,3.0)node[blacknode](t4ibwdgmdc9y){}; 
\draw (2.58481,3.05461)node[whitenode](ot4ibwdgmdc9y){};
\draw [move] (ot4ibwdgmdc9y) to (t4ibwdgmdc9y);
\draw (1.0,3.0)node[blacknode](1jy89i2mgzp7n){}; 
\draw (0.7365,3.0968)node[whitenode](o1jy89i2mgzp7n){};
\draw [move] (o1jy89i2mgzp7n) to (1jy89i2mgzp7n);
\draw (1.0,2.0)node[blacknode](u855jh32h6jq){}; 
\draw (0.685863,1.738)node[whitenode](ou855jh32h6jq){};
\draw [move] (ou855jh32h6jq) to (u855jh32h6jq);
\draw (2.0,1.0)node[blacknode](uv95saxzemgi){}; 
\draw (2.44134,1.41729)node[whitenode](ouv95saxzemgi){};
\draw [move] (ouv95saxzemgi) to (uv95saxzemgi);
\draw (4.0,0.0)node[blacknode](rlamfuxju6ut){}; 
\draw (3.68199,0.809619)node[whitenode](orlamfuxju6ut){};
\draw [move] (orlamfuxju6ut) to (rlamfuxju6ut);
\draw (4.0,2.0)node[blacknode](1iun6zwmozr5g){}; 
\draw (3.92674,2.27814)node[whitenode](o1iun6zwmozr5g){};
\draw [move] (o1iun6zwmozr5g) to (1iun6zwmozr5g);
\draw (3.0,1.0)node[blacknode](qhl1h74c4rj9){}; 
\draw (2.92241,1.81396)node[whitenode](oqhl1h74c4rj9){};
\draw [move] (oqhl1h74c4rj9) to (qhl1h74c4rj9);
\draw (3.0,2.0)node[blacknode](r17sdm2o0jtv){}; 
\draw (3.40347,1.98275)node[whitenode](or17sdm2o0jtv){};
\draw [move] (or17sdm2o0jtv) to (r17sdm2o0jtv);
\draw (t4ibwdgmdc9y) -- (to2amp8wb2as); 
\draw (1jy89i2mgzp7n) -- (t4ibwdgmdc9y); 
\draw (u855jh32h6jq) -- (1jy89i2mgzp7n); 
\draw (t4ibwdgmdc9y) -- (u855jh32h6jq); 
\draw (uv95saxzemgi) -- (t4ibwdgmdc9y); 
\draw (rlamfuxju6ut) -- (uv95saxzemgi); 
\draw (1iun6zwmozr5g) -- (rlamfuxju6ut); 
\draw (t4ibwdgmdc9y) -- (1iun6zwmozr5g); 
\draw (qhl1h74c4rj9) -- (t4ibwdgmdc9y); 
\draw (r17sdm2o0jtv) -- (rlamfuxju6ut); 
 \end{tikzreduction}
}\\

Input  && Output \\
\end{tabular}

\medskip

\centering
\begin{tabular}{l@{\quad}r@{\quad}r@{\quad}r@{\quad}r}
\toprule
      & \multicolumn{2}{c}{Full model}&\multicolumn{2}{c}{Row generation}\\
\midrule
Model size & rows & columns & rows & columns \\
\midrule
Without presolve & 12809 & 3631 & 3795 & 1133 \\
With presolve & 8046 & 2239 & 3791 & 1053 \\
\toprule
      & first solution & optimal 	& \multicolumn{2}{r}{optimal}\\
\midrule
Computation time & 3.2 s		& 90.6 s	& \multicolumn{2}{r}{29.2 s}\\ 
\bottomrule
\end{tabular}
\caption{Graph 1: $|V| = 9, |E| = 10$}
\label{example:graph_1}
\end{figure}

In Fig.~\ref{example:graph_2}, the by far most expensive constraint is for checking the embedding of the central node.
Any other constraint is easily satisfiable.
In terms of computation time, there is not a big difference between
finding the first solution and closing the integrality gap.
Notice, that every vertex has one preferred integer grid point that is not preferred by any other vertex, so just rounding to the nearest grid point already gives a optimal solution.
This solution is found by the first run of the row generation approach almost immediately.
(Note that the $0.5$ second runtime includes setting up the Java environment, calling the CPLEX solver and checking topology.)
The difference between the full model and the row generation approach becomes even more important when the size of the bounding box increases.
Consider Fig.~\ref{example:graph_3}. 
While still easy to round in the same sense as above, as size of the bounding rectangle increases, so does the time for building and solving the full model.
However, this has no impact on solving time for the row generation approach for ``easy'' graphs.

\begin{figure}[tb]
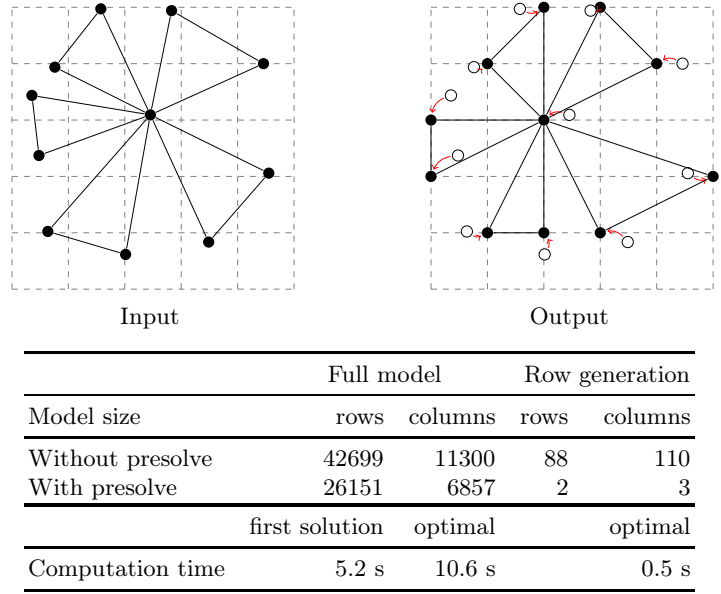

\centering
\begin{tabular}{cm{.1\textwidth}c}
\scalebox{.75}{
\begin{tikzreduction}{5}{5}
\draw (2.45281,3.09187)node[blacknode](qdsqzieux8vs){}; 
\draw (0.635306,1.02368)node[blacknode](skpkoiwtirlg){}; 
\draw (2.01411,0.616308)node[blacknode](tol5yc1bnacy){}; 
\draw (3.48691,0.835662)node[blacknode](qhnajvlfdilh){}; 
\draw (4.55235,2.05777)node[blacknode](1iyq1ycej5ois){}; 
\draw (4.45834,4.00063)node[blacknode](s4v4qszjlve9){}; 
\draw (2.82885,4.94072)node[blacknode](uvjqyl66tfmr){}; 
\draw (1.5754,4.97206)node[blacknode](uvbf8opegxut){}; 
\draw (0.760656,3.93796)node[blacknode](rkziyke8546e){}; 
\draw (0.353282,3.43657)node[blacknode](qxr42yna0rhv){}; 
\draw (0.478627,2.37114)node[blacknode](t4yhup5ul5cx){}; 
\draw (skpkoiwtirlg) -- (qdsqzieux8vs); 
\draw (tol5yc1bnacy) -- (skpkoiwtirlg); 
\draw (qdsqzieux8vs) -- (tol5yc1bnacy); 
\draw (qhnajvlfdilh) -- (qdsqzieux8vs); 
\draw (1iyq1ycej5ois) -- (qhnajvlfdilh); 
\draw (qdsqzieux8vs) -- (1iyq1ycej5ois); 
\draw (s4v4qszjlve9) -- (qdsqzieux8vs); 
\draw (uvjqyl66tfmr) -- (s4v4qszjlve9); 
\draw (qdsqzieux8vs) -- (uvjqyl66tfmr); 
\draw (uvbf8opegxut) -- (qdsqzieux8vs); 
\draw (rkziyke8546e) -- (uvbf8opegxut); 
\draw (qdsqzieux8vs) -- (rkziyke8546e); 
\draw (qxr42yna0rhv) -- (qdsqzieux8vs); 
\draw (t4yhup5ul5cx) -- (qxr42yna0rhv); 
\draw (qdsqzieux8vs) -- (t4yhup5ul5cx); 
\end{tikzreduction}
}
 & &
	\scalebox{.75}{
 \begin{tikzreduction}{5}{5}
  \draw (2.0,3.0)node[blacknode](s0x9gkay0htj){}; 
\draw (2.45281,3.09187)node[whitenode](os0x9gkay0htj){};
\draw [move] (os0x9gkay0htj) to (s0x9gkay0htj);
\draw (1.0,1.0)node[blacknode](qhnufc4fugxd){}; 
\draw (0.635306,1.02368)node[whitenode](oqhnufc4fugxd){};
\draw [move] (oqhnufc4fugxd) to (qhnufc4fugxd);
\draw (2.0,1.0)node[blacknode](1jei81te2flgm){}; 
\draw (2.01411,0.616308)node[whitenode](o1jei81te2flgm){};
\draw [move] (o1jei81te2flgm) to (1jei81te2flgm);
\draw (3.0,1.0)node[blacknode](1hrubiqvowhtk){}; 
\draw (3.48691,0.835662)node[whitenode](o1hrubiqvowhtk){};
\draw [move] (o1hrubiqvowhtk) to (1hrubiqvowhtk);
\draw (5.0,2.0)node[blacknode](qhnutjj9gs4k){}; 
\draw (4.55235,2.05777)node[whitenode](oqhnutjj9gs4k){};
\draw [move] (oqhnutjj9gs4k) to (qhnutjj9gs4k);
\draw (4.0,4.0)node[blacknode](1hv39joxlf588){}; 
\draw (4.45834,4.00063)node[whitenode](o1hv39joxlf588){};
\draw [move] (o1hv39joxlf588) to (1hv39joxlf588);
\draw (3.0,5.0)node[blacknode](urdzxtflaa0n){}; 
\draw (2.82885,4.94072)node[whitenode](ourdzxtflaa0n){};
\draw [move] (ourdzxtflaa0n) to (urdzxtflaa0n);
\draw (2.0,5.0)node[blacknode](vb9m1grtrjg9){}; 
\draw (1.5754,4.97206)node[whitenode](ovb9m1grtrjg9){};
\draw [move] (ovb9m1grtrjg9) to (vb9m1grtrjg9);
\draw (1.0,4.0)node[blacknode](s1jgidsbs5rm){}; 
\draw (0.760656,3.93796)node[whitenode](os1jgidsbs5rm){};
\draw [move] (os1jgidsbs5rm) to (s1jgidsbs5rm);
\draw (0.0,3.0)node[blacknode](t84kwvv1qmw3){}; 
\draw (0.353282,3.43657)node[whitenode](ot84kwvv1qmw3){};
\draw [move] (ot84kwvv1qmw3) to (t84kwvv1qmw3);
\draw (0.0,2.0)node[blacknode](vf97oiip875d){}; 
\draw (0.478627,2.37114)node[whitenode](ovf97oiip875d){};
\draw [move] (ovf97oiip875d) to (vf97oiip875d);
\draw (qhnufc4fugxd) -- (s0x9gkay0htj); 
\draw (1jei81te2flgm) -- (qhnufc4fugxd); 
\draw (s0x9gkay0htj) -- (1jei81te2flgm); 
\draw (1hrubiqvowhtk) -- (s0x9gkay0htj); 
\draw (qhnutjj9gs4k) -- (1hrubiqvowhtk); 
\draw (s0x9gkay0htj) -- (qhnutjj9gs4k); 
\draw (1hv39joxlf588) -- (s0x9gkay0htj); 
\draw (urdzxtflaa0n) -- (1hv39joxlf588); 
\draw (s0x9gkay0htj) -- (urdzxtflaa0n); 
\draw (vb9m1grtrjg9) -- (s0x9gkay0htj); 
\draw (s1jgidsbs5rm) -- (vb9m1grtrjg9); 
\draw (s0x9gkay0htj) -- (s1jgidsbs5rm); 
\draw (t84kwvv1qmw3) -- (s0x9gkay0htj); 
\draw (vf97oiip875d) -- (t84kwvv1qmw3); 
\draw (s0x9gkay0htj) -- (vf97oiip875d); 
 \end{tikzreduction}
}\\

Input && Output \\
\end{tabular}

\medskip

\centering
\begin{tabular}{l@{\quad}r@{\quad}r@{\quad}r@{\quad}r}
\toprule
      & \multicolumn{2}{c}{Full model}&\multicolumn{2}{c}{Row generation}\\
\midrule
Model size & rows & columns & rows & columns \\
\midrule
Without presolve & 42699 & 11300 & 88 & 110 \\
With presolve & 26151 &6857 & 2 & 3 \\
\toprule
      & first solution & optimal 	& \multicolumn{2}{r}{optimal}\\
\midrule
Computation time & 5.2 s		& 10.6 s	& \multicolumn{2}{r}{0.5 s}\\ 
\bottomrule
\end{tabular}
\caption{Graph 2: $|V| = 11, |E| = 15$}
\label{example:graph_2}
\end{figure}

Consider the graph in Fig.~\ref{example:graph_4}.
While too large to round with the full model approach in the allotted time of $10$ minutes, the row generation only has to add one constraint from Equation \ref{equation:constraint-overlap} for two vertices of the upper-right corner.
Rebuilding the model and solving with this constraint runs in reasonable time (compared to the full model).
Notice that this constraint does not involve the direction set $\mathcal{D}$.

When rounding graphs with vertices starting in close proximity (like in Fig.~\ref{example:graph_5}) several things can be noticed.
First of all, small bounding box result in small and easy-to-solve models.
The size of the bounding box has extreme effect on the runtime (compare Fig.~\ref{example:graph_1}, which has only two more vertices but a much larger bounding box).
Second, when many constraints are violated during the row generation processes, iteratively adding the constraints results in runtime exceeding the time for solving the full model in the first place.

We end this section with two rather small examples
(Figs.~\ref{example:graph_5} and~\ref{example:graph_6}).
Both have a rather large number of vertices compared to the size of the bounding rectangle and thus include many ``difficult'' parts.
The row generation approach clearly outperforms the full model (while
still being infeasible in practice).

\begin{figure}[bp]
\centering
\begin{tabular}{cm{.02\textwidth}c}
\resizebox{0.25\textwidth}{!}{
\begin{tikzreduction}{27}{26}
\draw (1.32304,1.70636)node[blacknode](somycxbeshmg){}; 
\draw (25.6204,1.78032)node[blacknode](ureklq192ue9){}; 
\draw (16.4118,5.73742)node[blacknode](skk1lqyt0g1d){}; 
\draw (6.05678,5.73742)node[blacknode](1hruao1z7oawi){}; 
\draw (1.50796,20.5303)node[blacknode](vey5r8ji7p7l){}; 
\draw (5.94583,15.8336)node[blacknode](toanil2gz67o){}; 
\draw (15.6722,15.3158)node[blacknode](vb42i0n3lflg){}; 
\draw (16.1529,20.1605)node[blacknode](u7okgi4mklyp){}; 
\draw (23.7343,15.1309)node[blacknode](qe10rcxpdt13){}; 
\draw (26.0272,23.23)node[blacknode](rl2ulwg4llcy){}; 
\draw (1.39701,24.8942)node[blacknode](1iemp6keto8qd){}; 
\draw (1.69287,21.6398)node[blacknode](1jy7r41e8bf3s){}; 
\draw (24.5479,23.9697)node[blacknode](qe9v6h3bm7uf){}; 
\draw [thick] (ureklq192ue9) -- (somycxbeshmg); 
\draw [thick] (skk1lqyt0g1d) -- (ureklq192ue9); 
\draw [thick] (1hruao1z7oawi) -- (skk1lqyt0g1d); 
\draw [thick] (somycxbeshmg) -- (1hruao1z7oawi); 
\draw [thick] (vey5r8ji7p7l) -- (somycxbeshmg); 
\draw [thick] (toanil2gz67o) -- (vey5r8ji7p7l); 
\draw [thick] (1hruao1z7oawi) -- (toanil2gz67o); 
\draw [thick] (vb42i0n3lflg) -- (1hruao1z7oawi); 
\draw [thick] (toanil2gz67o) -- (vb42i0n3lflg); 
\draw [thick] (u7okgi4mklyp) -- (toanil2gz67o); 
\draw [thick] (vb42i0n3lflg) -- (u7okgi4mklyp); 
\draw [thick] (skk1lqyt0g1d) -- (vb42i0n3lflg); 
\draw [thick] (qe10rcxpdt13) -- (skk1lqyt0g1d); 
\draw [thick] (vb42i0n3lflg) -- (qe10rcxpdt13); 
\draw [thick] (rl2ulwg4llcy) -- (vb42i0n3lflg); 
\draw [thick] (ureklq192ue9) -- (rl2ulwg4llcy); 
\draw [thick] (qe10rcxpdt13) -- (ureklq192ue9); 
\draw [thick] (u7okgi4mklyp) -- (vey5r8ji7p7l); 
\draw [thick] (rl2ulwg4llcy) -- (u7okgi4mklyp); 
\draw [thick] (vey5r8ji7p7l) -- (rl2ulwg4llcy); 
\draw [thick] (1jy7r41e8bf3s) -- (1iemp6keto8qd); 
\draw [thick] (qe9v6h3bm7uf) -- (1jy7r41e8bf3s); 
\draw [thick] (1iemp6keto8qd) -- (qe9v6h3bm7uf); 
\draw [thick] (rl2ulwg4llcy) -- (qe9v6h3bm7uf); 
\draw [thick] (vey5r8ji7p7l) -- (1jy7r41e8bf3s); 
\end{tikzreduction}
}
 &&
 \resizebox{0.715\textwidth}{!}{
 \begin{tikzreduction}{27}{26}
  \draw (1.0,2.0)node[blacknode](1jyaikj75op4i){}; 
\draw (1.32304,1.70636)node[whitenode](o1jyaikj75op4i){};
\draw [move] (o1jyaikj75op4i) to (1jyaikj75op4i);
\draw (26.0,2.0)node[blacknode](sksue3m747y1){}; 
\draw (25.6204,1.78032)node[whitenode](osksue3m747y1){};
\draw [move] (osksue3m747y1) to (sksue3m747y1);
\draw (16.0,6.0)node[blacknode](1iauwjxir73uf){}; 
\draw (16.4118,5.73742)node[whitenode](o1iauwjxir73uf){};
\draw [move] (o1iauwjxir73uf) to (1iauwjxir73uf);
\draw (6.0,6.0)node[blacknode](skk17pum5f05){}; 
\draw (6.05678,5.73742)node[whitenode](oskk17pum5f05){};
\draw [move] (oskk17pum5f05) to (skk17pum5f05);
\draw (2.0,21.0)node[blacknode](1ibjs90fzby9h){}; 
\draw (1.50796,20.5303)node[whitenode](o1ibjs90fzby9h){};
\draw [move] (o1ibjs90fzby9h) to (1ibjs90fzby9h);
\draw (6.0,16.0)node[blacknode](qxi7q34mtnc5){}; 
\draw (5.94583,15.8336)node[whitenode](oqxi7q34mtnc5){};
\draw [move] (oqxi7q34mtnc5) to (qxi7q34mtnc5);
\draw (16.0,15.0)node[blacknode](t49j1dozm5o1){}; 
\draw (15.6722,15.3158)node[whitenode](ot49j1dozm5o1){};
\draw [move] (ot49j1dozm5o1) to (t49j1dozm5o1);
\draw (16.0,20.0)node[blacknode](t4ncqrt1xfe8){}; 
\draw (16.1529,20.1605)node[whitenode](ot4ncqrt1xfe8){};
\draw [move] (ot4ncqrt1xfe8) to (t4ncqrt1xfe8);
\draw (24.0,15.0)node[blacknode](sl66yaiqscip){}; 
\draw (23.7343,15.1309)node[whitenode](osl66yaiqscip){};
\draw [move] (osl66yaiqscip) to (sl66yaiqscip);
\draw (26.0,23.0)node[blacknode](ubjly2p0vpdk){}; 
\draw (26.0272,23.23)node[whitenode](oubjly2p0vpdk){};
\draw [move] (oubjly2p0vpdk) to (ubjly2p0vpdk);
\draw (1.0,25.0)node[blacknode](1k1zhepiw22bc){}; 
\draw (1.39701,24.8942)node[whitenode](o1k1zhepiw22bc){};
\draw [move] (o1k1zhepiw22bc) to (1k1zhepiw22bc);
\draw (2.0,22.0)node[blacknode](1ibesw5rwx7q9){}; 
\draw (1.69287,21.6398)node[whitenode](o1ibesw5rwx7q9){};
\draw [move] (o1ibesw5rwx7q9) to (1ibesw5rwx7q9);
\draw (25.0,24.0)node[blacknode](qy4xsqkucig1){}; 
\draw (24.5479,23.9697)node[whitenode](oqy4xsqkucig1){};
\draw [move] (oqy4xsqkucig1) to (qy4xsqkucig1);
\draw [thick] (sksue3m747y1) -- (1jyaikj75op4i); 
\draw [thick] (1iauwjxir73uf) -- (sksue3m747y1); 
\draw [thick] (skk17pum5f05) -- (1iauwjxir73uf); 
\draw [thick] (1jyaikj75op4i) -- (skk17pum5f05); 
\draw [thick] (1ibjs90fzby9h) -- (1jyaikj75op4i); 
\draw [thick] (qxi7q34mtnc5) -- (1ibjs90fzby9h); 
\draw [thick] (skk17pum5f05) -- (qxi7q34mtnc5); 
\draw [thick] (t49j1dozm5o1) -- (skk17pum5f05); 
\draw [thick] (qxi7q34mtnc5) -- (t49j1dozm5o1); 
\draw [thick] (t4ncqrt1xfe8) -- (qxi7q34mtnc5); 
\draw [thick] (t49j1dozm5o1) -- (t4ncqrt1xfe8); 
\draw [thick] (1iauwjxir73uf) -- (t49j1dozm5o1); 
\draw [thick] (sl66yaiqscip) -- (1iauwjxir73uf); 
\draw [thick] (t49j1dozm5o1) -- (sl66yaiqscip); 
\draw [thick] (ubjly2p0vpdk) -- (t49j1dozm5o1); 
\draw [thick] (sksue3m747y1) -- (ubjly2p0vpdk); 
\draw [thick] (sl66yaiqscip) -- (sksue3m747y1); 
\draw [thick] (t4ncqrt1xfe8) -- (1ibjs90fzby9h); 
\draw [thick] (ubjly2p0vpdk) -- (t4ncqrt1xfe8); 
\draw [thick] (1ibjs90fzby9h) -- (ubjly2p0vpdk); 
\draw [thick] (1ibesw5rwx7q9) -- (1k1zhepiw22bc); 
\draw [thick] (qy4xsqkucig1) -- (1ibesw5rwx7q9); 
\draw [thick] (1k1zhepiw22bc) -- (qy4xsqkucig1); 
\draw [thick] (ubjly2p0vpdk) -- (qy4xsqkucig1); 
\draw [thick] (1ibjs90fzby9h) -- (1ibesw5rwx7q9); 
 \end{tikzreduction}
 }\\

Input && Output \\
\end{tabular}

\medskip

\centering
\begin{tabular}{l@{\quad}r@{\quad}r@{\quad}r@{\quad}r}
\toprule
      & \multicolumn{2}{c}{Full model}&\multicolumn{2}{c}{Row generation}\\
\midrule
Model size & rows & columns & rows & columns \\
\midrule
Without presolve & * & * & 104 & 130 \\
With presolve  & * & * & 4 & 5 \\
\toprule
      & first solution & optimal 	& \multicolumn{2}{r}{optimal}\\
\midrule
Computation time & \dag	& \dag	& \multicolumn{2}{r}{0.4 s}	\\ 
\bottomrule
\end{tabular}
\caption{Graph 3: $|V| = 13, |E| = 25$}
\label{example:graph_3}
\end{figure}

\begin{figure}[tp]
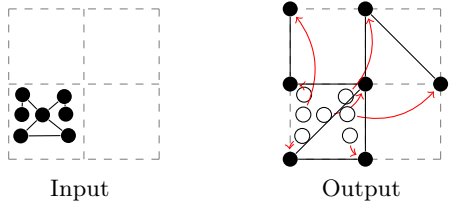

\centering
\begin{tabular}{cm{.1\textwidth}c}
\resizebox{0.45\textwidth}{!}{
\begin{tikzreduction}{14}{14}
\draw (0.794544,0.80155)node[blacknode](u7zo9czx7c69){}; 
\draw (2.61047,2.08096)node[blacknode](1ib0g65xwl8oi){}; 
\draw (5.16928,1.04917)node[blacknode](1h82kk7ky7u48){}; 
\draw (6.61375,2.34922)node[blacknode](toa2jhz210s9){}; 
\draw (2.77556,3.381)node[blacknode](1k1x8nv73ab8h){}; 
\draw (0.650094,6.06362)node[blacknode](to2bt2t3on5d){}; 
\draw (4.22004,6.53825)node[blacknode](t89ovl2wvxv6){}; 
\draw (6.24233,3.58735)node[blacknode](t4q7ba847k2s){}; 
\draw (6.82012,6.90969)node[blacknode](qe1kpijtvqzm){}; 
\draw (8.45031,2.34922)node[blacknode](to4l9e7apzjm){}; 
\draw (9.44088,6.47631)node[blacknode](1jefyzgrc6gqe){}; 
\draw (9.64719,0.946)node[blacknode](s4p2hbwi7x9w){}; 
\draw (3.27081,5.48583)node[blacknode](ubgaohj2ausj){}; 
\draw (2.92001,4.37151)node[blacknode](tonykyqxasky){}; 
\draw (4.50894,4.28896)node[blacknode](u7p2phtf6uk3){}; 
\draw (1.39063,11.2733)node[blacknode](uvhgnayfdhki){}; 
\draw (1.44428,7.81269)node[blacknode](1hbg1fg120svn){}; 
\draw (9.03613,8.537)node[blacknode](1hrotj5n42o8o){}; 
\draw (11.1822,3.06444)node[blacknode](urpnozfs6reh){}; 
\draw (10.6189,11.7025)node[blacknode](t87d2udbuc12){}; 
\draw (0.800456,13.3121)node[blacknode](qxtub3ua6xx0){}; 
\draw (0.773625,11.8902)node[blacknode](soem95oaxxdj){}; 
\draw (10.9944,12.1049)node[blacknode](1h7geas8d2b74){}; 
\draw (12.0138,0.542771)node[blacknode](qxtdxhsvym9i){}; 
\draw (13.3283,0.6769)node[blacknode](1iepgykiv67n8){}; 
\draw (12.3626,13.6876)node[blacknode](u88h77ao6cfa){}; 
\draw (1ib0g65xwl8oi) -- (u7zo9czx7c69); 
\draw (1h82kk7ky7u48) -- (1ib0g65xwl8oi); 
\draw (toa2jhz210s9) -- (1h82kk7ky7u48); 
\draw (1k1x8nv73ab8h) -- (toa2jhz210s9); 
\draw (to2bt2t3on5d) -- (1k1x8nv73ab8h); 
\draw (t89ovl2wvxv6) -- (to2bt2t3on5d); 
\draw (t4q7ba847k2s) -- (t89ovl2wvxv6); 
\draw (qe1kpijtvqzm) -- (t4q7ba847k2s); 
\draw (to4l9e7apzjm) -- (qe1kpijtvqzm); 
\draw (1jefyzgrc6gqe) -- (to4l9e7apzjm); 
\draw (s4p2hbwi7x9w) -- (1jefyzgrc6gqe); 
\draw (t4q7ba847k2s) -- (1k1x8nv73ab8h); 
\draw (to4l9e7apzjm) -- (t4q7ba847k2s); 
\draw (s4p2hbwi7x9w) -- (to4l9e7apzjm); 
\draw (s4p2hbwi7x9w) -- (toa2jhz210s9); 
\draw (1h82kk7ky7u48) -- (s4p2hbwi7x9w); 
\draw (tonykyqxasky) -- (ubgaohj2ausj); 
\draw (u7p2phtf6uk3) -- (tonykyqxasky); 
\draw (ubgaohj2ausj) -- (u7p2phtf6uk3); 
\draw (1hbg1fg120svn) -- (uvhgnayfdhki); 
\draw (qe1kpijtvqzm) -- (1hbg1fg120svn); 
\draw (1hrotj5n42o8o) -- (qe1kpijtvqzm); 
\draw (urpnozfs6reh) -- (1hrotj5n42o8o); 
\draw (t87d2udbuc12) -- (urpnozfs6reh); 
\draw (1hrotj5n42o8o) -- (t87d2udbuc12); 
\draw (uvhgnayfdhki) -- (1hrotj5n42o8o); 
\draw (qe1kpijtvqzm) -- (uvhgnayfdhki); 
\draw (soem95oaxxdj) -- (qxtub3ua6xx0); 
\draw (1h7geas8d2b74) -- (soem95oaxxdj); 
\draw (qxtdxhsvym9i) -- (1h7geas8d2b74); 
\draw (1iepgykiv67n8) -- (qxtdxhsvym9i); 
\draw (u88h77ao6cfa) -- (1iepgykiv67n8); 
\draw (qxtub3ua6xx0) -- (u88h77ao6cfa); 
\draw (u88h77ao6cfa) -- (1h7geas8d2b74); 
\end{tikzreduction}
}
 &&
 \resizebox{0.45\textwidth}{!}{
 \begin{tikzreduction}{14}{14}
  \draw (1.0,1.0)node[blacknode](to2bvbaw7m1y){}; 
\draw (0.794544,0.80155)node[whitenode](oto2bvbaw7m1y){};
\draw [move] (oto2bvbaw7m1y) to (to2bvbaw7m1y);
\draw (3.0,2.0)node[blacknode](qxljr4sz1s8h){}; 
\draw (2.61047,2.08096)node[whitenode](oqxljr4sz1s8h){};
\draw [move] (oqxljr4sz1s8h) to (qxljr4sz1s8h);
\draw (5.0,1.0)node[blacknode](1k22bg7a8fwvq){}; 
\draw (5.16928,1.04917)node[whitenode](o1k22bg7a8fwvq){};
\draw [move] (o1k22bg7a8fwvq) to (1k22bg7a8fwvq);
\draw (7.0,2.0)node[blacknode](skvl040elap2){}; 
\draw (6.61375,2.34922)node[whitenode](oskvl040elap2){};
\draw [move] (oskvl040elap2) to (skvl040elap2);
\draw (3.0,3.0)node[blacknode](skk2dr8dnui9){}; 
\draw (2.77556,3.381)node[whitenode](oskk2dr8dnui9){};
\draw [move] (oskk2dr8dnui9) to (skk2dr8dnui9);
\draw (1.0,6.0)node[blacknode](1h7rftrmmu7ug){}; 
\draw (0.650094,6.06362)node[whitenode](o1h7rftrmmu7ug){};
\draw [move] (o1h7rftrmmu7ug) to (1h7rftrmmu7ug);
\draw (4.0,7.0)node[blacknode](1hux59pt3o0v9){}; 
\draw (4.22004,6.53825)node[whitenode](o1hux59pt3o0v9){};
\draw [move] (o1hux59pt3o0v9) to (1hux59pt3o0v9);
\draw (6.0,4.0)node[blacknode](1hb7mzdrn5q3o){}; 
\draw (6.24233,3.58735)node[whitenode](o1hb7mzdrn5q3o){};
\draw [move] (o1hb7mzdrn5q3o) to (1hb7mzdrn5q3o);
\draw (7.0,7.0)node[blacknode](vbf6dm41ya0g){}; 
\draw (6.82012,6.90969)node[whitenode](ovbf6dm41ya0g){};
\draw [move] (ovbf6dm41ya0g) to (vbf6dm41ya0g);
\draw (8.0,2.0)node[blacknode](rkwt4zmyj9ro){}; 
\draw (8.45031,2.34922)node[whitenode](orkwt4zmyj9ro){};
\draw [move] (orkwt4zmyj9ro) to (rkwt4zmyj9ro);
\draw (9.0,6.0)node[blacknode](r1ifkmxd94j7){}; 
\draw (9.44088,6.47631)node[whitenode](or1ifkmxd94j7){};
\draw [move] (or1ifkmxd94j7) to (r1ifkmxd94j7);
\draw (10.0,1.0)node[blacknode](1iesa7y1ysj76){}; 
\draw (9.64719,0.946)node[whitenode](o1iesa7y1ysj76){};
\draw [move] (o1iesa7y1ysj76) to (1iesa7y1ysj76);
\draw (3.0,5.0)node[blacknode](1hv0gx9pyo484){}; 
\draw (3.27081,5.48583)node[whitenode](o1hv0gx9pyo484){};
\draw [move] (o1hv0gx9pyo484) to (1hv0gx9pyo484);
\draw (3.0,4.0)node[blacknode](qe9ch1uau2ch){}; 
\draw (2.92001,4.37151)node[whitenode](oqe9ch1uau2ch){};
\draw [move] (oqe9ch1uau2ch) to (qe9ch1uau2ch);
\draw (5.0,4.0)node[blacknode](1je9xyvzbka45){}; 
\draw (4.50894,4.28896)node[whitenode](o1je9xyvzbka45){};
\draw [move] (o1je9xyvzbka45) to (1je9xyvzbka45);
\draw (1.0,11.0)node[blacknode](qhw4yxfb0ew5){}; 
\draw (1.39063,11.2733)node[whitenode](oqhw4yxfb0ew5){};
\draw [move] (oqhw4yxfb0ew5) to (qhw4yxfb0ew5);
\draw (1.0,8.0)node[blacknode](rl2vgftpqng4){}; 
\draw (1.44428,7.81269)node[whitenode](orl2vgftpqng4){};
\draw [move] (orl2vgftpqng4) to (rl2vgftpqng4);
\draw (9.0,9.0)node[blacknode](u8b8kvq48tv6){}; 
\draw (9.03613,8.537)node[whitenode](ou8b8kvq48tv6){};
\draw [move] (ou8b8kvq48tv6) to (u8b8kvq48tv6);
\draw (11.0,3.0)node[blacknode](1jifjgjj8mhbn){}; 
\draw (11.1822,3.06444)node[whitenode](o1jifjgjj8mhbn){};
\draw [move] (o1jifjgjj8mhbn) to (1jifjgjj8mhbn);
\draw (11.0,11.0)node[blacknode](1iunpbk9xf037){}; 
\draw (10.6189,11.7025)node[whitenode](o1iunpbk9xf037){};
\draw [move] (o1iunpbk9xf037) to (1iunpbk9xf037);
\draw (1.0,13.0)node[blacknode](t4nv4ih5zfp0){}; 
\draw (0.800456,13.3121)node[whitenode](ot4nv4ih5zfp0){};
\draw [move] (ot4nv4ih5zfp0) to (t4nv4ih5zfp0);
\draw (1.0,12.0)node[blacknode](1jyouzk1q4ilx){}; 
\draw (0.773625,11.8902)node[whitenode](o1jyouzk1q4ilx){};
\draw [move] (o1jyouzk1q4ilx) to (1jyouzk1q4ilx);
\draw (11.0,12.0)node[blacknode](vbcea603iwkh){}; 
\draw (10.9944,12.1049)node[whitenode](ovbcea603iwkh){};
\draw [move] (ovbcea603iwkh) to (vbcea603iwkh);
\draw (12.0,1.0)node[blacknode](1jicsuj4ushgh){}; 
\draw (12.0138,0.542771)node[whitenode](o1jicsuj4ushgh){};
\draw [move] (o1jicsuj4ushgh) to (1jicsuj4ushgh);
\draw (13.0,1.0)node[blacknode](1hre8akymywc9){}; 
\draw (13.3283,0.6769)node[whitenode](o1hre8akymywc9){};
\draw [move] (o1hre8akymywc9) to (1hre8akymywc9);
\draw (12.0,14.0)node[blacknode](sl6q1mwsdank){}; 
\draw (12.3626,13.6876)node[whitenode](osl6q1mwsdank){};
\draw [move] (osl6q1mwsdank) to (sl6q1mwsdank);
\draw (1k22bg7a8fwvq) -- (qxljr4sz1s8h) -- (to2bvbaw7m1y); 
\draw (1hux59pt3o0v9) -- (1h7rftrmmu7ug) -- (skk2dr8dnui9) -- (skvl040elap2) -- (1k22bg7a8fwvq); 
\draw (rkwt4zmyj9ro) -- (vbf6dm41ya0g) -- (1hb7mzdrn5q3o) -- (1hux59pt3o0v9); 
\draw (1iesa7y1ysj76) -- (r1ifkmxd94j7) -- (rkwt4zmyj9ro); 
\draw (1iesa7y1ysj76) -- (rkwt4zmyj9ro) -- (1hb7mzdrn5q3o) -- (skk2dr8dnui9); 
\draw (1k22bg7a8fwvq) -- (1iesa7y1ysj76) -- (skvl040elap2); 
\draw (1hv0gx9pyo484) -- (1je9xyvzbka45) -- (qe9ch1uau2ch) -- (1hv0gx9pyo484); 
\draw (qhw4yxfb0ew5) -- (rl2vgftpqng4) -- (vbf6dm41ya0g) -- (qhw4yxfb0ew5) -- (u8b8kvq48tv6) -- (1iunpbk9xf037) -- (1jifjgjj8mhbn) -- (u8b8kvq48tv6) -- (vbf6dm41ya0g);
\draw (1jicsuj4ushgh) -- (vbcea603iwkh) -- (1jyouzk1q4ilx) -- (t4nv4ih5zfp0); 
\draw (sl6q1mwsdank) -- (1hre8akymywc9) -- (1jicsuj4ushgh); 
\draw (t4nv4ih5zfp0) -- (sl6q1mwsdank) -- (vbcea603iwkh); 
 \end{tikzreduction}
 }\\

Input && Output \\
\end{tabular}

\medskip

\centering
\begin{tabular}{l@{\quad}r@{\quad}r@{\quad}r@{\quad}r}
\toprule
      & \multicolumn{2}{c}{Full model}&\multicolumn{2}{c}{Row generation}\\
\midrule
Model size & rows & columns & rows & columns \\
\midrule
Without presolve & * & * & 24762 & 6482\\
With presolve & * & * & 12355 & 3146 \\
\toprule
      & first solution & optimal 	& \multicolumn{2}{r}{optimal}\\
\midrule
Computation time & \dag	& \dag	& \multicolumn{2}{r}{7.1 s}\\ 
\bottomrule
\end{tabular}
\caption{Graph 4: $|V| = 26, |E| = 34$}
\label{example:graph_4}
\end{figure}

\begin{figure}[pt]
\centering
\begin{tabular}{cm{.1\textwidth}c}
\begin{tikzreduction}{2}{2}
\draw (0.184991,0.859075)node[blacknode](us06zqsrcjmp){}; 
\draw (0.44967,0.589401)node[blacknode](1jed61jlernmq){}; 
\draw (0.79425,0.314734)node[blacknode](1hb7oj8coo4s7){}; 
\draw (0.160021,0.30974)node[blacknode](1iyhsi3fxa1ix){}; 
\draw (0.739319,0.834106)node[blacknode](1ibc0cfbkt0s4){}; 
\draw (0.175003,0.594395)node[blacknode](ubrf9el6e3c3){}; 
\draw (0.759294,0.599389)node[blacknode](1jetbja7pc4l0){}; 
\draw (1jed61jlernmq) -- (us06zqsrcjmp); 
\draw (1hb7oj8coo4s7) -- (1jed61jlernmq); 
\draw (1iyhsi3fxa1ix) -- (1hb7oj8coo4s7); 
\draw (1jed61jlernmq) -- (1iyhsi3fxa1ix); 
\draw (1ibc0cfbkt0s4) -- (1jed61jlernmq); 
\draw (us06zqsrcjmp) -- (ubrf9el6e3c3); 
\draw (1ibc0cfbkt0s4) -- (1jetbja7pc4l0); 
\end{tikzreduction}
 & &
 \begin{tikzreduction}{2}{2}
 \draw (0.0,1.0)node[blacknode](us06zqsrcjmp){}; 
\draw (0.184991,0.859075)node[whitenode](ous06zqsrcjmp){};
\draw [move] (ous06zqsrcjmp) to (us06zqsrcjmp);
\draw (1.0,1.0)node[blacknode](1jed61jlernmq){}; 
\draw (0.44967,0.589401)node[whitenode](o1jed61jlernmq){};
\draw [move] (o1jed61jlernmq) to (1jed61jlernmq);
\draw (1.0,0.0)node[blacknode](1hb7oj8coo4s7){}; 
\draw (0.79425,0.314734)node[whitenode](o1hb7oj8coo4s7){};
\draw [move] (o1hb7oj8coo4s7) to (1hb7oj8coo4s7);
\draw (0.0,0.0)node[blacknode](1iyhsi3fxa1ix){}; 
\draw (0.160021,0.30974)node[whitenode](o1iyhsi3fxa1ix){};
\draw [move] (o1iyhsi3fxa1ix) to (1iyhsi3fxa1ix);
\draw (1.0,2.0)node[blacknode](1ibc0cfbkt0s4){}; 
\draw (0.739319,0.834106)node[whitenode](o1ibc0cfbkt0s4){};
\draw [move] (o1ibc0cfbkt0s4) to (1ibc0cfbkt0s4);
\draw (0.0,2.0)node[blacknode](ubrf9el6e3c3){}; 
\draw (0.175003,0.594395)node[whitenode](oubrf9el6e3c3){};
\draw [move] (oubrf9el6e3c3) to (ubrf9el6e3c3);
\draw (2.0,1.0)node[blacknode](1jetbja7pc4l0){}; 
\draw (0.759294,0.599389)node[whitenode](o1jetbja7pc4l0){};
\draw [move] (o1jetbja7pc4l0) to (1jetbja7pc4l0);
\draw (1iyhsi3fxa1ix) -- (1hb7oj8coo4s7) -- (1jed61jlernmq) -- (us06zqsrcjmp);
\draw (1jetbja7pc4l0) -- (1ibc0cfbkt0s4) -- (1jed61jlernmq) -- (1iyhsi3fxa1ix); 
\draw (us06zqsrcjmp) -- (ubrf9el6e3c3); 
 \end{tikzreduction}\\

Input && Output \\
\end{tabular}

\medskip
\centering
\begin{tabular}{l@{\quad}r@{\quad}r@{\quad}r@{\quad}r}
\toprule
      & \multicolumn{2}{c}{Full model}&\multicolumn{2}{c}{Row generation}\\
\midrule
Model size & rows & columns & rows & columns \\
\midrule
Without presolve & 2603 & 916 & 2271 & 816 \\
With presolve & 2583 & 896 & 2245 & 682 \\
\toprule
      & first solution & optimal 	& \multicolumn{2}{r}{optimal}\\
\midrule
Computation time & 0.5 s		& 4.8 s	& \multicolumn{2}{r}{20.2 s} \\ 
\bottomrule
\end{tabular}
\caption{Graph 5: $|V| = 7, |E| = 7$}
\label{example:graph_5}
\end{figure}

\begin{figure}[pt]
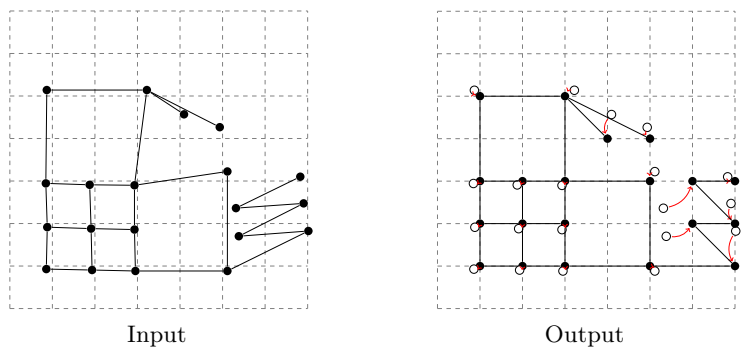

\centering
\begin{tabular}{cm{.1\textwidth}c}
 \resizebox{0.3\textwidth}{!}{
 \begin{tikzreduction}{5}{7}
\draw (0.274571,0.379452)node[blacknode](1h82ixih2ezxg){}; 
\draw (0.408984,3.75096)node[blacknode](sl8zi8wtlzcl){}; 
\draw (1.09224,0.513864)node[blacknode](qy4dwxsdqzw6){}; 
\draw (1.33867,2.17161)node[blacknode](1h7fv6yeqvinr){}; 
\draw (2.14514,0.480261)node[blacknode](1jyg1nqmbpa9c){}; 
\draw (2.34676,3.71735)node[blacknode](1jew4h5ffaxip){}; 
\draw (3.06362,0.782687)node[blacknode](qdw0b5k28svq){}; 
\draw (3.35485,1.70117)node[blacknode](s1jfnuf2jinp){}; 
\draw (4.16132,0.525065)node[blacknode](sobs7tfygmyg){}; 
\draw (4.41894,3.66134)node[blacknode](vbq8d9ci8yt0){}; 
\draw (3.62367,6.39438)node[blacknode](1iemqcrouh839){}; 
\draw (3.48926,2.48524)node[blacknode](t4kkndnckfn9){}; 
\draw (2.78359,6.28238)node[blacknode](vbt1bef3l7hu){}; 
\draw (2.29075,4.32221)node[blacknode](1iupz89ycbo85){}; 
\draw (1.61869,6.34956)node[blacknode](1k22acyw39f1e){}; 
\draw (1.57389,3.30291)node[blacknode](to2bsmodhik4){}; 
\draw (0.756212,6.39438)node[blacknode](1ibh3wvgmpik1){}; 
\draw (0.476189,4.50142)node[blacknode](vbknyqf399b9){}; 
\draw (0.308174,6.32719)node[blacknode](1ji7r5rnl23xg){}; 
\draw (sl8zi8wtlzcl) -- (1h82ixih2ezxg); 
\draw (qy4dwxsdqzw6) -- (sl8zi8wtlzcl); 
\draw (1h7fv6yeqvinr) -- (qy4dwxsdqzw6); 
\draw (1jyg1nqmbpa9c) -- (1h7fv6yeqvinr); 
\draw (1jew4h5ffaxip) -- (1jyg1nqmbpa9c); 
\draw (qdw0b5k28svq) -- (1jew4h5ffaxip); 
\draw (s1jfnuf2jinp) -- (qdw0b5k28svq); 
\draw (sobs7tfygmyg) -- (s1jfnuf2jinp); 
\draw (vbq8d9ci8yt0) -- (sobs7tfygmyg); 
\draw (1iemqcrouh839) -- (vbq8d9ci8yt0); 
\draw (t4kkndnckfn9) -- (1iemqcrouh839); 
\draw (vbt1bef3l7hu) -- (t4kkndnckfn9); 
\draw (1iupz89ycbo85) -- (vbt1bef3l7hu); 
\draw (1k22acyw39f1e) -- (1iupz89ycbo85); 
\draw (to2bsmodhik4) -- (1k22acyw39f1e); 
\draw (1ibh3wvgmpik1) -- (to2bsmodhik4); 
\draw (vbknyqf399b9) -- (1ibh3wvgmpik1); 
\draw (1ji7r5rnl23xg) -- (vbknyqf399b9); 
\end{tikzreduction}
 }
 &&
  \resizebox{0.3\textwidth}{!}{
 \begin{tikzreduction}{5}{7}
  \draw (0.0,0.0)node[blacknode](1h82ixih2ezxg){}; 
\draw (0.274571,0.379452)node[whitenode](o1h82ixih2ezxg){};
\draw [move] (o1h82ixih2ezxg) to (1h82ixih2ezxg);
\draw (0.0,4.0)node[blacknode](sl8zi8wtlzcl){}; 
\draw (0.408984,3.75096)node[whitenode](osl8zi8wtlzcl){};
\draw [move] (osl8zi8wtlzcl) to (sl8zi8wtlzcl);
\draw (1.0,1.0)node[blacknode](qy4dwxsdqzw6){}; 
\draw (1.09224,0.513864)node[whitenode](oqy4dwxsdqzw6){};
\draw [move] (oqy4dwxsdqzw6) to (qy4dwxsdqzw6);
\draw (1.0,2.0)node[blacknode](1h7fv6yeqvinr){}; 
\draw (1.33867,2.17161)node[whitenode](o1h7fv6yeqvinr){};
\draw [move] (o1h7fv6yeqvinr) to (1h7fv6yeqvinr);
\draw (2.0,0.0)node[blacknode](1jyg1nqmbpa9c){}; 
\draw (2.14514,0.480261)node[whitenode](o1jyg1nqmbpa9c){};
\draw [move] (o1jyg1nqmbpa9c) to (1jyg1nqmbpa9c);
\draw (2.0,4.0)node[blacknode](1jew4h5ffaxip){}; 
\draw (2.34676,3.71735)node[whitenode](o1jew4h5ffaxip){};
\draw [move] (o1jew4h5ffaxip) to (1jew4h5ffaxip);
\draw (3.0,1.0)node[blacknode](qdw0b5k28svq){}; 
\draw (3.06362,0.782687)node[whitenode](oqdw0b5k28svq){};
\draw [move] (oqdw0b5k28svq) to (qdw0b5k28svq);
\draw (3.0,2.0)node[blacknode](s1jfnuf2jinp){}; 
\draw (3.35485,1.70117)node[whitenode](os1jfnuf2jinp){};
\draw [move] (os1jfnuf2jinp) to (s1jfnuf2jinp);
\draw (4.0,1.0)node[blacknode](sobs7tfygmyg){}; 
\draw (4.16132,0.525065)node[whitenode](osobs7tfygmyg){};
\draw [move] (osobs7tfygmyg) to (sobs7tfygmyg);
\draw (4.0,4.0)node[blacknode](vbq8d9ci8yt0){}; 
\draw (4.41894,3.66134)node[whitenode](ovbq8d9ci8yt0){};
\draw [move] (ovbq8d9ci8yt0) to (vbq8d9ci8yt0);
\draw (4.0,6.0)node[blacknode](1iemqcrouh839){}; 
\draw (3.62367,6.39438)node[whitenode](o1iemqcrouh839){};
\draw [move] (o1iemqcrouh839) to (1iemqcrouh839);
\draw (3.0,3.0)node[blacknode](t4kkndnckfn9){}; 
\draw (3.48926,2.48524)node[whitenode](ot4kkndnckfn9){};
\draw [move] (ot4kkndnckfn9) to (t4kkndnckfn9);
\draw (3.0,6.0)node[blacknode](vbt1bef3l7hu){}; 
\draw (2.78359,6.28238)node[whitenode](ovbt1bef3l7hu){};
\draw [move] (ovbt1bef3l7hu) to (vbt1bef3l7hu);
\draw (2.0,5.0)node[blacknode](1iupz89ycbo85){}; 
\draw (2.29075,4.32221)node[whitenode](o1iupz89ycbo85){};
\draw [move] (o1iupz89ycbo85) to (1iupz89ycbo85);
\draw (2.0,6.0)node[blacknode](1k22acyw39f1e){}; 
\draw (1.61869,6.34956)node[whitenode](o1k22acyw39f1e){};
\draw [move] (o1k22acyw39f1e) to (1k22acyw39f1e);
\draw (1.0,3.0)node[blacknode](to2bsmodhik4){}; 
\draw (1.57389,3.30291)node[whitenode](oto2bsmodhik4){};
\draw [move] (oto2bsmodhik4) to (to2bsmodhik4);
\draw (1.0,6.0)node[blacknode](1ibh3wvgmpik1){}; 
\draw (0.756212,6.39438)node[whitenode](o1ibh3wvgmpik1){};
\draw [move] (o1ibh3wvgmpik1) to (1ibh3wvgmpik1);
\draw (0.0,5.0)node[blacknode](vbknyqf399b9){}; 
\draw (0.476189,4.50142)node[whitenode](ovbknyqf399b9){};
\draw [move] (ovbknyqf399b9) to (vbknyqf399b9);
\draw (0.0,6.0)node[blacknode](1ji7r5rnl23xg){}; 
\draw (0.308174,6.32719)node[whitenode](o1ji7r5rnl23xg){};
\draw [move] (o1ji7r5rnl23xg) to (1ji7r5rnl23xg);
\draw (sl8zi8wtlzcl) -- (1h82ixih2ezxg); 
\draw (qy4dwxsdqzw6) -- (sl8zi8wtlzcl); 
\draw (1h7fv6yeqvinr) -- (qy4dwxsdqzw6); 
\draw (1jyg1nqmbpa9c) -- (1h7fv6yeqvinr); 
\draw (1jew4h5ffaxip) -- (1jyg1nqmbpa9c); 
\draw (qdw0b5k28svq) -- (1jew4h5ffaxip); 
\draw (s1jfnuf2jinp) -- (qdw0b5k28svq); 
\draw (sobs7tfygmyg) -- (s1jfnuf2jinp); 
\draw (vbq8d9ci8yt0) -- (sobs7tfygmyg); 
\draw (1iemqcrouh839) -- (vbq8d9ci8yt0); 
\draw (t4kkndnckfn9) -- (1iemqcrouh839); 
\draw (vbt1bef3l7hu) -- (t4kkndnckfn9); 
\draw (1iupz89ycbo85) -- (vbt1bef3l7hu); 
\draw (1k22acyw39f1e) -- (1iupz89ycbo85); 
\draw (to2bsmodhik4) -- (1k22acyw39f1e); 
\draw (1ibh3wvgmpik1) -- (to2bsmodhik4); 
\draw (vbknyqf399b9) -- (1ibh3wvgmpik1); 
\draw (1ji7r5rnl23xg) -- (vbknyqf399b9); 
 \end{tikzreduction}
  }
  \\

Input && Output \\
\end{tabular}

\medskip

\centering
\begin{tabular}{l@{\quad}r@{\quad}r@{\quad}r@{\quad}r}
\toprule
      & \multicolumn{2}{c}{Full model}&\multicolumn{2}{c}{Row generation}\\
\midrule
Model size & rows & columns & rows & columns \\
\midrule
Without presolve & 135386 & 35649& 15741 & 4174 \\
With presolve & 74957 & 19591 & 9200 & 2402 \\
\toprule
      & first solution & optimal 	& \multicolumn{2}{r}{optimal}\\
\midrule
Computation time & 42.6 s	& 1105.6 s	& \multicolumn{2}{r}{21.2 s}\\ 
\bottomrule
\end{tabular}
\caption{Graph 6: $|V| = 19, |E| = 18$}
\label{example:graph_6}
\end{figure}

\begin{figure}[pt]
\centering
\begin{tabular}{cm{.1\textwidth}c}
\resizebox{0.35\textwidth}{!}{
\begin{tikzreduction}{7}{7}
\draw (0.857025,0.9283)node[blacknode](rl558lemo9d1){}; 
\draw (1.93232,0.9059)node[blacknode](trtkvv1gb86v){}; 
\draw (1.90992,1.88039)node[blacknode](qhsuvnscvof5){}; 
\draw (0.879425,1.91399)node[blacknode](t4vqsatoi4vn){}; 
\draw (0.845825,2.94448)node[blacknode](s12tusz400o2){}; 
\draw (1.87631,2.91088)node[blacknode](1iyl0bvgelqgk){}; 
\draw (2.95161,0.8835)node[blacknode](r1a3e67xpf78){}; 
\draw (2.92921,1.85798)node[blacknode](trwelsfm3v4x){}; 
\draw (2.92921,2.89967)node[blacknode](skk1ll1xohmg){}; 
\draw (5.11341,0.8835)node[blacknode](1h7j4ilkq90me){}; 
\draw (5.11341,3.22451)node[blacknode](tocwl3nmtrnb){}; 
\draw (0.868225,5.13987)node[blacknode](1jyg0k8dittno){}; 
\draw (3.22044,5.13987)node[blacknode](1ibkbnx6lzv3n){}; 
\draw (4.93419,4.2662)node[blacknode](toogu1nnhqp0){}; 
\draw (4.09411,4.56862)node[blacknode](1ibe9h8o76tyx){}; 
\draw (7.01756,1.82438)node[blacknode](vbccnkq63w37){}; 
\draw (5.38223,1.70117)node[blacknode](s15n4lyl1a2e){}; 
\draw (6.90556,2.47404)node[blacknode](1jydrf530zeqr){}; 
\draw (5.31503,2.36203)node[blacknode](1k1wqcmzmf4gz){}; 
\draw (6.82719,3.10129)node[blacknode](rhm4fkmua6nr){}; 
\draw (trtkvv1gb86v) -- (rl558lemo9d1); 
\draw (qhsuvnscvof5) -- (trtkvv1gb86v); 
\draw (t4vqsatoi4vn) -- (qhsuvnscvof5); 
\draw (rl558lemo9d1) -- (t4vqsatoi4vn); 
\draw (s12tusz400o2) -- (t4vqsatoi4vn); 
\draw (1iyl0bvgelqgk) -- (s12tusz400o2); 
\draw (qhsuvnscvof5) -- (1iyl0bvgelqgk); 
\draw (r1a3e67xpf78) -- (trtkvv1gb86v); 
\draw (trwelsfm3v4x) -- (r1a3e67xpf78); 
\draw (qhsuvnscvof5) -- (trwelsfm3v4x); 
\draw (skk1ll1xohmg) -- (1iyl0bvgelqgk); 
\draw (trwelsfm3v4x) -- (skk1ll1xohmg); 
\draw (1h7j4ilkq90me) -- (r1a3e67xpf78); 
\draw (tocwl3nmtrnb) -- (1h7j4ilkq90me); 
\draw (skk1ll1xohmg) -- (tocwl3nmtrnb); 
\draw (1jyg0k8dittno) -- (s12tusz400o2); 
\draw (1ibkbnx6lzv3n) -- (1jyg0k8dittno); 
\draw (skk1ll1xohmg) -- (1ibkbnx6lzv3n); 
\draw (toogu1nnhqp0) -- (1ibkbnx6lzv3n); 
\draw (1ibe9h8o76tyx) -- (1ibkbnx6lzv3n); 
\draw (vbccnkq63w37) -- (1h7j4ilkq90me); 
\draw (s15n4lyl1a2e) -- (vbccnkq63w37); 
\draw (1jydrf530zeqr) -- (s15n4lyl1a2e); 
\draw (1k1wqcmzmf4gz) -- (1jydrf530zeqr); 
\draw (rhm4fkmua6nr) -- (1k1wqcmzmf4gz); 
\end{tikzreduction}
}
 &&
 \resizebox{0.35\textwidth}{!}{
 \begin{tikzreduction}{7}{7}
  \draw (1.0,1.0)node[blacknode](rl558lemo9d1){}; 
\draw (0.857025,0.9283)node[whitenode](orl558lemo9d1){};
\draw [move] (orl558lemo9d1) to (rl558lemo9d1);
\draw (2.0,1.0)node[blacknode](trtkvv1gb86v){}; 
\draw (1.93232,0.9059)node[whitenode](otrtkvv1gb86v){};
\draw [move] (otrtkvv1gb86v) to (trtkvv1gb86v);
\draw (2.0,2.0)node[blacknode](qhsuvnscvof5){}; 
\draw (1.90992,1.88039)node[whitenode](oqhsuvnscvof5){};
\draw [move] (oqhsuvnscvof5) to (qhsuvnscvof5);
\draw (1.0,2.0)node[blacknode](t4vqsatoi4vn){}; 
\draw (0.879425,1.91399)node[whitenode](ot4vqsatoi4vn){};
\draw [move] (ot4vqsatoi4vn) to (t4vqsatoi4vn);
\draw (1.0,3.0)node[blacknode](s12tusz400o2){}; 
\draw (0.845825,2.94448)node[whitenode](os12tusz400o2){};
\draw [move] (os12tusz400o2) to (s12tusz400o2);
\draw (2.0,3.0)node[blacknode](1iyl0bvgelqgk){}; 
\draw (1.87631,2.91088)node[whitenode](o1iyl0bvgelqgk){};
\draw [move] (o1iyl0bvgelqgk) to (1iyl0bvgelqgk);
\draw (3.0,1.0)node[blacknode](r1a3e67xpf78){}; 
\draw (2.95161,0.8835)node[whitenode](or1a3e67xpf78){};
\draw [move] (or1a3e67xpf78) to (r1a3e67xpf78);
\draw (3.0,2.0)node[blacknode](trwelsfm3v4x){}; 
\draw (2.92921,1.85798)node[whitenode](otrwelsfm3v4x){};
\draw [move] (otrwelsfm3v4x) to (trwelsfm3v4x);
\draw (3.0,3.0)node[blacknode](skk1ll1xohmg){}; 
\draw (2.92921,2.89967)node[whitenode](oskk1ll1xohmg){};
\draw [move] (oskk1ll1xohmg) to (skk1ll1xohmg);
\draw (5.0,1.0)node[blacknode](1h7j4ilkq90me){}; 
\draw (5.11341,0.8835)node[whitenode](o1h7j4ilkq90me){};
\draw [move] (o1h7j4ilkq90me) to (1h7j4ilkq90me);
\draw (5.0,3.0)node[blacknode](tocwl3nmtrnb){}; 
\draw (5.11341,3.22451)node[whitenode](otocwl3nmtrnb){};
\draw [move] (otocwl3nmtrnb) to (tocwl3nmtrnb);
\draw (1.0,5.0)node[blacknode](1jyg0k8dittno){}; 
\draw (0.868225,5.13987)node[whitenode](o1jyg0k8dittno){};
\draw [move] (o1jyg0k8dittno) to (1jyg0k8dittno);
\draw (3.0,5.0)node[blacknode](1ibkbnx6lzv3n){}; 
\draw (3.22044,5.13987)node[whitenode](o1ibkbnx6lzv3n){};
\draw [move] (o1ibkbnx6lzv3n) to (1ibkbnx6lzv3n);
\draw (5.0,4.0)node[blacknode](toogu1nnhqp0){}; 
\draw (4.93419,4.2662)node[whitenode](otoogu1nnhqp0){};
\draw [move] (otoogu1nnhqp0) to (toogu1nnhqp0);
\draw (4.0,4.0)node[blacknode](1ibe9h8o76tyx){}; 
\draw (4.09411,4.56862)node[whitenode](o1ibe9h8o76tyx){};
\draw [move] (o1ibe9h8o76tyx) to (1ibe9h8o76tyx);
\draw (7.0,1.0)node[blacknode](vbccnkq63w37){}; 
\draw (7.01756,1.82438)node[whitenode](ovbccnkq63w37){};
\draw [move] (ovbccnkq63w37) to (vbccnkq63w37);
\draw (6.0,2.0)node[blacknode](s15n4lyl1a2e){}; 
\draw (5.38223,1.70117)node[whitenode](os15n4lyl1a2e){};
\draw [move] (os15n4lyl1a2e) to (s15n4lyl1a2e);
\draw (7.0,2.0)node[blacknode](1jydrf530zeqr){}; 
\draw (6.90556,2.47404)node[whitenode](o1jydrf530zeqr){};
\draw [move] (o1jydrf530zeqr) to (1jydrf530zeqr);
\draw (6.0,3.0)node[blacknode](1k1wqcmzmf4gz){}; 
\draw (5.31503,2.36203)node[whitenode](o1k1wqcmzmf4gz){};
\draw [move] (o1k1wqcmzmf4gz) to (1k1wqcmzmf4gz);
\draw (7.0,3.0)node[blacknode](rhm4fkmua6nr){}; 
\draw (6.82719,3.10129)node[whitenode](orhm4fkmua6nr){};
\draw [move] (orhm4fkmua6nr) to (rhm4fkmua6nr);
\draw (trtkvv1gb86v) -- (rl558lemo9d1); 
\draw (qhsuvnscvof5) -- (trtkvv1gb86v); 
\draw (t4vqsatoi4vn) -- (qhsuvnscvof5); 
\draw (rl558lemo9d1) -- (t4vqsatoi4vn); 
\draw (s12tusz400o2) -- (t4vqsatoi4vn); 
\draw (1iyl0bvgelqgk) -- (s12tusz400o2); 
\draw (qhsuvnscvof5) -- (1iyl0bvgelqgk); 
\draw (r1a3e67xpf78) -- (trtkvv1gb86v); 
\draw (trwelsfm3v4x) -- (r1a3e67xpf78); 
\draw (qhsuvnscvof5) -- (trwelsfm3v4x); 
\draw (skk1ll1xohmg) -- (1iyl0bvgelqgk); 
\draw (trwelsfm3v4x) -- (skk1ll1xohmg); 
\draw (1h7j4ilkq90me) -- (r1a3e67xpf78); 
\draw (tocwl3nmtrnb) -- (1h7j4ilkq90me); 
\draw (skk1ll1xohmg) -- (tocwl3nmtrnb); 
\draw (1jyg0k8dittno) -- (s12tusz400o2); 
\draw (1ibkbnx6lzv3n) -- (1jyg0k8dittno); 
\draw (skk1ll1xohmg) -- (1ibkbnx6lzv3n); 
\draw (toogu1nnhqp0) -- (1ibkbnx6lzv3n); 
\draw (1ibe9h8o76tyx) -- (1ibkbnx6lzv3n); 
\draw (vbccnkq63w37) -- (1h7j4ilkq90me); 
\draw (s15n4lyl1a2e) -- (vbccnkq63w37); 
\draw (1jydrf530zeqr) -- (s15n4lyl1a2e); 
\draw (1k1wqcmzmf4gz) -- (1jydrf530zeqr); 
\draw (rhm4fkmua6nr) -- (1k1wqcmzmf4gz); 
 \end{tikzreduction}
 }\\

Input && Output \\
\end{tabular}

\medskip

\centering
\begin{tabular}{l@{\quad}r@{\quad}r@{\quad}r@{\quad}r}
\toprule
      & \multicolumn{2}{c}{Full model}&\multicolumn{2}{c}{Row generation}\\
\midrule
Model size & rows & columns & rows & columns \\
\midrule
Without presolve & 323441 & 82816 & 15161 & 4044 \\
With presolve & 323441 & 82816 & 15127 & 3894 \\
\toprule
      & first solution & optimal 	& \multicolumn{2}{r}{optimal}\\
\midrule
Computation time & 182.1 s	& \dag  	& \multicolumn{2}{r}{211.6 s}\\ 
\bottomrule
\end{tabular}
\caption{Graph 7: $|V| = 20, |E| = 25$}
\label{example:graph_7}

\end{figure}

\else
\fi

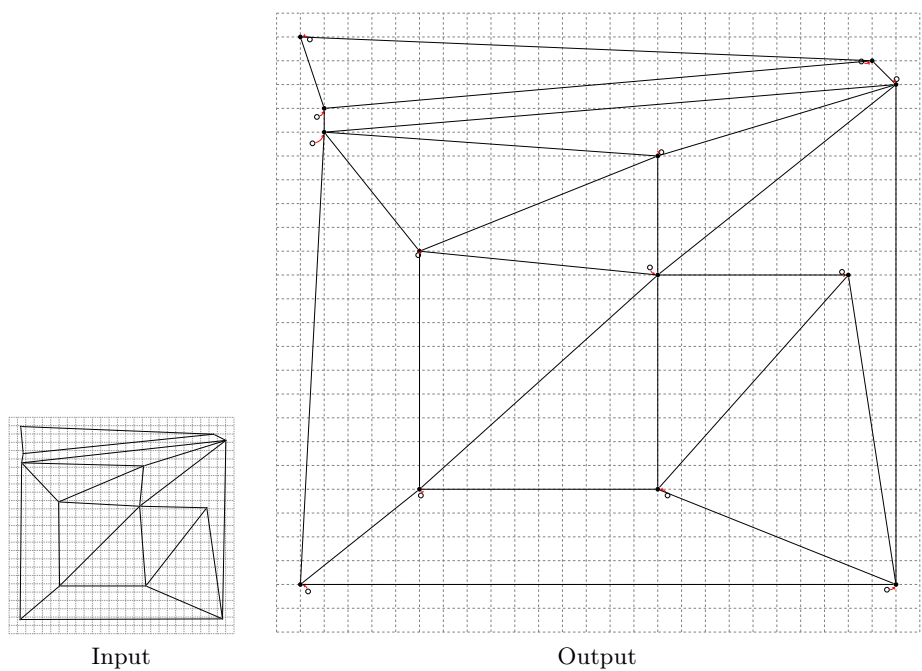
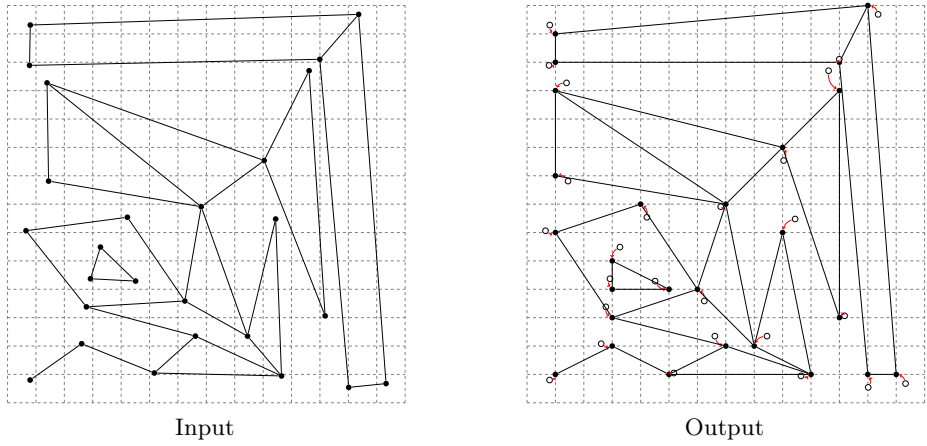
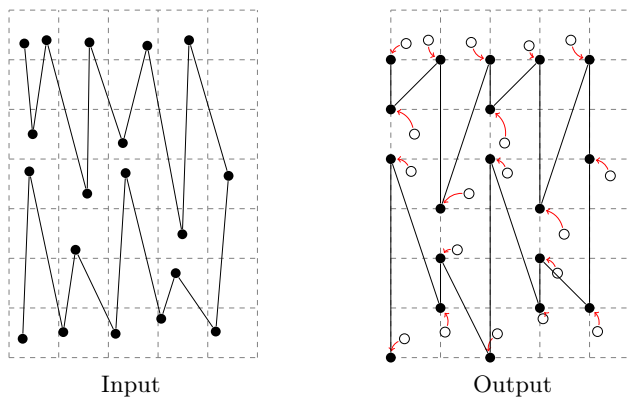
\end{document}